\documentclass[12pt]{amsart}
\usepackage[]{graphicx}
\usepackage{lineno}
\usepackage{graphicx}
\usepackage{amsmath}
\usepackage{amsfonts}
\usepackage{natbib}
\usepackage{hyperref}
\hypersetup{
    colorlinks=true,
    linkcolor=blue,
    filecolor=magenta,
    urlcolor=cyan,
    citecolor=blue,
}

\usepackage{color}

\newtheorem{thm}{Theorem}[section]

\newtheorem{cor}{Corollary}[section]
\newtheorem{lem}{Lemma}[section]
\newtheorem{prop}{Proposition}[section]
\theoremstyle{definition}

\theoremstyle{remark}

\numberwithin{equation}{section}

\newcommand{\s}{\mathcal{S}}

\newcommand{\E}{\mathbb{E}}

\newcommand{\R}{\mathbb{R}}

\renewcommand{\P}{\mathbb{P}}

\numberwithin{equation}{section}

\newcommand{\bed}{\begin{displaymath}}
\newcommand{\eed}{\end{displaymath}}
\newcommand{\bea}{\bed\begin{array}{rl}}
\newcommand{\eea}{\end{array}\eed}

\newcommand{\barray}{\begin{array}{ll}}
\newcommand{\earray}{\end{array}}

\def\bar{\overline}
\def\hat{\widehat}
\def\a.s{\text{\;a.s.\;}}

\begin{document}
\bibliographystyle{plainnat}

\title[Persistence and extinction for stochastic difference equations]{Persistence and extinction for stochastic ecological models with internal and external variables}

\author[M. Bena\"{\i}m]{Michel Bena\"{\i}m}
\address{Institut de Math\'{e}matiques, Universit\'{e} de Neuch\^{a}tel, Rue Emile-Argand, Neuch\^{a}tel, Suisse-2000.}
\author[S.J. Schreiber]{Sebastian J. Schreiber$^*$}\footnotetext{$^*$Corresponding author: \href{mailto:sschreiber@ucdavis.edu}{sschreiber@ucdavis.edu}}
\address{Department of Evolution and Ecology and Center for Population Biology, University of California, Davis, California USA 95616}
\maketitle

\textbf{Abstract.}  The dynamics of species' densities depend both on internal and external variables. Internal variables include frequencies of individuals exhibiting different phenotypes or living in different spatial locations. External variables include abiotic factors or non-focal species. These internal or external variables may fluctuate due to stochastic fluctuations in environmental conditions.  The interplay between these variables and species densities can determine whether a particular population persists or goes extinct. To understand this interplay, we prove theorems for stochastic persistence and exclusion for stochastic  ecological difference equations accounting for internal and external variables. Specifically, we use a stochastic analog of average Lyapunov functions to develop sufficient and necessary conditions for (i) all population densities spending little time at low densities i.e. stochastic persistence, and (ii) population trajectories asymptotically approaching the extinction set with positive probability. For (i) and (ii), respectively, we provide quantitative estimates on the fraction of time that the system is near the extinction set, and the probability of asymptotic extinction as a function of the initial state of the system. Furthermore, in the case of persistence, we provide lower bounds for the expected time to escape neighborhoods of the extinction set. To illustrate the applicability of our results, we analyze stochastic models of evolutionary games, Lotka-Volterra dynamics, trait evolution, and spatially structured disease dynamics. Our analysis of these models demonstrates environmental stochasticity facilitates coexistence of strategies in the hawk-dove game, but inhibits coexistence in the rock-paper-scissors game and a Lotka-Volterra predator-prey model. Furthermore, environmental fluctuations with positive auto-correlations can promote persistence of evolving populations and persistence of diseases in patchy landscapes. While our results help close the gap between the persistence theories for deterministic and stochastic systems, we highlight several challenges for future research.

%\newpage
\section{Introduction}

In population biology, environmental stochasticity refers to the effects of fluctuations in environmental factors (e.g. temperature, precipitation) on demography. These demographic effects include fluctuations in survival, growth, and reproduction and can result in fluctuations in population densities, disease prevalence, and genotypic frequencies. This stochasticity can drive populations extinct~\citep{lewontin-cohen-69,gyllenberg-etal-94b,mclaughlin-etal-02,tpb-09,hening2017stochastic}, facilitate coexistence~\citep{hutchinson-61,chesson-warner-81,chesson-82,chesson-85,chesson-ellner-89,chesson-94,kuang-chesson-09,chesson-18}, reverse competitive dominance~\citep{benaim-lobry-16},  maintain or disrupt genetic diversity~\citep{gillespie-73,gillespie-78,gillespie-turelli-89}, and alter the persistence and spread of infectious diseases~\citep{altizer-etal-06}.

One approach to studying these effects is to analyze  stochastic difference or differential equations. For these equations, stochastic persistence corresponds to a statistical tendency for the population dynamics to avoid low densities~\citep{chesson-82,chesson-ellner-89,jdea-11}. Thus far, general methods for identifying stochastic persistence exist for  unstructured populations~\citep{jmb-11,hening-nguyen-18}, and structured (e.g. spatial or stage structure) discrete-time models on compact state spaces~\citep{jmb-14}. Few general results, however, exist for identifying when one or more species are asymptotically tending toward extinction. A notable exception is the work of \citet{hening-nguyen-18} who proved a general condition for extinction for stochastic differential equation models of unstructured populations. This condition, however, does not apply to models with intransitive outcomes, such as the evolutionary game of rock-paper-scissor~\citep{hofbauer-sigmund-03}. Moreover, none of these aforementioned persistence or extinction results apply to models with general forms of internal and external variables~\citep{jmb-18}.

All populations have internal and external variables that impact their population dynamics and, more generally, the ecological dynamics of the communities in which they reside. Internal variables are intrinsic to the population. For example, most populations exhibit a multitude of genotypes which vary in demographically important traits. Differential survival and reproduction of the genotypes alter the distributions of these traits. These frequency changes within the population may alter the rate at which the population grows creating a feedback between the population density  and the trait distribution (the internal variable)~\citep{vincent-brown-05,schoener-11,amnat-18a}. Similarly, for populations structured by age or space, the distribution of ages or spatial locations are internal variables that often generate feedbacks with the total population density~\citep{chesson-00a,jmb-13,hening-etal-18}.

External variables are extrinsic to the population. For example, environmental conditions (e.g. precipitation, temperature, nitrogen availability) often fluctuate over time and impact survivorship, growth, and reproduction of individuals. Indeed, the traditional view of environmental stochasticity corresponds to fluctuations in these external variables driving ecological dynamics, but not vice-versa. However, there is  evidence of active, bidirectional feedbacks between the abundance and composition of ecological communities and local weather patterns~\citep{eltahir-98,zeng-etal-99,kucharski-etal-13} and fire~\citep{staver-levin-12}. These bidirectional feedbacks also arise in ecological communities with ecosystem engineers, such as beavers and oysters, that modulate the environment in a way that influences their demography, as well as, other species in the community~\citep{jones-etal-94,cuddington-etal-09}. External variables also may include other species that interact with the focal species of a given model or system. 

As these internal and external variables may influence persistence and extinction, we develop criteria for these outcomes in stochastic difference equations allowing for these auxiliary variables. In section~\ref{sec:model}, we introduce the class of multi-species stochastic difference equations with auxiliary variables and our main assumptions. In section~\ref{sec:persist}, we discuss two concepts of stochastic persistence (almost-sure stochastic persistence and stochastic persistence in probability) and define realized per-capita growth rates with respect to  invariant measures. To prove results about both forms of stochastic persistence, we use the methods of \citet{benaim-18} that introduces stochastic analogs of \citet{hofbauer-81}'s average Lyapunov functions. The proofs in our setting (compact and discrete-time) simplify substantially when compared to the general results of \citep{benaim-18} and, thereby, highlight the core ingredients of these new methods. In section~\ref{sec:extinction}, we prove two new classes of extinction results for stochastic difference equations. One of these results shows if there is a stochastically persisting subcommunity that can not be invaded by the other species, then the other species are extinction prone. This result is a discrete-time analog of a result proven by \citep{hening-nguyen-18} for stochastic differential equation models of unstructured interacting species. Our second extinction result naturally complements our persistence theorem and provides a sufficient condition for convergence with positive probability to the extinction set, where one or more species have density zero. Unlike the first extinction result, the second extinction result allows us to handle rock-paper-scissor type dynamics. In section~\ref{sec:apps}, we apply our results to models of evolutionary games, multi-species interactions of the Lotka-Volterra type,  trait evolution in a fluctuating environment, and spatially structured disease dynamics. Each of these examples is designed to highlight different ways that our results can be used. In section~\ref{sec:end}, we conclude with discussing how our results relate to the existing literature and raising  open mathematical challenges.

\section{Models and Assumptions}\label{sec:model} We consider $n$ interacting species whose densities at time $t$ are given by $X_t=(X_t^1,X_t^2,\dots,X_t^n)$ which lies in $[0,\infty)^n=:\R^n_+$. The species dynamics interact with the auxiliary variable $Y_t$ which lies in $(-\infty,\infty)^k=:\R^k$. As discussed below, $Y_t$ may correspond to abiotic forcing, feedbacks with environment variables, the internal structure of each species, or any combination thereof. Both of these variables may be influenced by stochastic forces captured by a sequence of independent and identically distributed (i.i.d.) random variables $\xi_1,\xi_2,\dots$ taking values in a Polish space $\Xi$ i.e. a separable completely metrizable topological space. The fitness $f^i$ of individuals in species $i$ depends on the densities $X_t$ of all species, the auxiliary variables $Y_t$, and the random variable $\xi_{t+1}$ that captures the stochastic changes that occur over the time interval $(t,t+1].$ Similarly, the update rule $G$ of the auxiliary variables depends on $X_t$, $Y_t$, and $\xi_{t+1}$.

Under these assumptions, our model is
\begin{equation}\label{eq:main}
\begin{aligned}
X^i_{t+1} &= X^i_t f^i(X_t,Y_t,\xi_{t+1}) \quad i=1,2,\dots,n& \mbox{ (species densities)}\\
Y_{t+1}&= G(X_t,Y_t,\xi_{t+1})& \mbox{ (auxiliary variables)}.
\end{aligned}
\end{equation}
Our standing assumptions for \eqref{eq:main} are:
\begin{description}
\item [A1] For each $i=1,2,\dots,n$, the fitness function $f^i(z,\xi)$ is continuous in $z=(x,y)$, measurable in $(z,\xi)$, and strictly positive.
\item [A2] The auxiliary variable update function $G$ is continuous in $z=(x,y)$ and measurable in $(z,\xi)$.
\item [A3] There is a compact subset $\s$ of $\R^n_+\times \R^k$ such that all solutions $Z_t=(X_t,Y_t)$ to \eqref{eq:main} satisfy $Z_t\in \s$ for $t$ sufficiently large. 
\item [A4] For all $i=1,2,\dots,n$, $\sup_{z,\xi}|\log f^i(z,\xi)|<\infty$.
\end{description}
Assumptions \textbf{A1}--\textbf{A2} ensure that the Markov chain $Z_t$ is Feller (see section~\ref{subsec:cgr} for a definition) and that positive densities are mapped to positive densities. Assumption \textbf{A3} implies the dynamics remain bounded. Assumption \textbf{A4} uniformly bounds the log fitness of all species. These assumptions are often met in models as illustrated in Section~\ref{sec:apps}.

Some examples of the auxiliary variable $Y_t$ include:
\begin{description}
	\item[Markovian environmental forcing] To model Markovian fluctuations in the environment, $Y_{t+1}$ can represent the state of the environment over the time interval $(t,t+1]$ which determines the fitness of the species. Many forms of Markovian dynamics can be represented as $Y_{t+1}=G(Y_t,\xi_{t+1})$ where $\xi_1,\xi_2,\dots$ are i.i.d. For example, autoregressive processes of the form $Y_{t+1}=A Y_t+\xi_{t+1}$ where $A$ is a matrix whose spectral radius is less than one. In this case, if $\xi_t$ take values in a compact set, then $Y_t$ converges to a unique stationary distribution which is compactly supported~\citep{diaconis-freedman-99}. Alternatively, suppose there a finite number of environmental states $\{1,2,\dots,k\}$ with transition probabilities $p^{ij}$ i.e. $\P[Y_{t+1}=j|Y_t=i]=p^{ij}$. One can represent this Markov chain $Y_t$ as a composition of random maps by defining  $\xi_t=(\xi_t^1,\dots,\xi_t^k)$ to be a random vector such that $\P[\xi_t^i=j]=p^{ij}$ and defining $G(Y,\xi)=\xi^{Y}$. The representation of Markov chains as random maps can be done more generally (e.g. continuous state spaces in Theorem 1.2 of \cite{kifer-86}). For any of these choices of the auxiliary variable dynamics, the fitness function are of the form $f^i(X_t,Y_{t+1})=f^i(X_t,G(Y_t,\xi_{t+1})).$
	\item[Population structure] Many forms of discrete population structure (e.g. age, stages, spatial) can be represented by \eqref{eq:main} in the $Y$ dynamics. For example, if the interacting species live in $k$ distinct patches, then the auxiliary variable $Y_t=(Y_t^{ij})_{\{1\le i \le n, 1\le j\le k\}}$ can be a matrix where $Y^{ij}_t$ is the fraction of species $i$ living in patch $j$. The dynamics of these spatial frequencies of species $i$, even in the limit $X_t^i=0$, can be derived directly from a model description of the species densities in each patch. 
	\item[Trait evolution] The auxiliary variables can represent ecologically-important traits under selection. Changes in these traits can drive changes in population densities and, conversely, the densities of the species can select for different trait values. One classic model of trait evolution is due to \citet{lande-76} in which the trait $Y^i$ of species $i$ follows the gradient of the log fitness of the species i.e. $Y^i_{t+1}=Y_t^i + \alpha \frac{\partial \log f^i(X_t,Y_t,\xi_{t+1})}{\partial Y^i}$ for some constant $\alpha>0$ that determines the speed of evolution. This gradient type dynamic also appears in models of adaptive dynamics~\citep{metz-etal-96,michod-00,vincent-brown-05}. More generally, these auxiliary variables may keep track of the frequencies of different genotypes and their associated trait distribution~\citep{gillespie-78,maladaptive-19}.
	\item[Environmental feedbacks for ecosystem engineers] Ecosystem engineers are organisms, such as beavers and oysters, that modulate the abiotic environment in a manner influencing the demography of itself or other species~\citep{jones-etal-94,cuddington-etal-09}. For example, many species of oysters produce large, complex reef structures which are composed of living and dead oysters. These reefs provide a positive feedback on oyster growth and survival, and provide ecosystem services, including  water filtration, habitat and predator refuge for a variety of other species~\citep{moore2016demographic,moore-etal-18}. To model these system dynamics of oysters, modelers include the reef size $Y_t$ as an auxiliary variable $Y_t$ that can substantially alter the dynamics of the system~\citep{jordan-etal-11,moore-etal-18}.
\end{description}
Several illustrative examples of these different uses of the auxiliary variable $Y_t$ are given in Section~\ref{sec:apps}.

\section{Stochastic persistence}\label{sec:persist}

In this section, we begin by stating two definitions of stochastic persistence~\citep{jdea-11}: persistence in probability and almost-sure persistence. The first corresponds to what \citet{chesson-82} called stochastically bounded coexistence and takes an ensemble point of view. The second takes the perspective of a single, typical realization of the Markov chain~\citep{jmb-11}. To derive a persistence criterion, we define the realized per-capita growth rates when species are infinitesimally rare. This approach to characterizing coexistence in stochastic models goes back to the work of \citet{turelli-81}. Using these realized per-capita growth rates, we introduce the stochastic analog of the realized community per-capita growth  rate when rare due to \citet{hofbauer-81}. This community growth rate allows us to define the stochastic analog of \citet{hofbauer-81}'s average Lyapunov functions and, thereby, prove sufficient conditions for stochastic persistence of either form.

\subsection{Persistence in probability and almost-sure persistence} As we just mentioned, there are two ways to think about the asymptotic behavior of $\{Z_t\}_{t=0}^\infty$. First, one can ask what is the distribution of $Z_t=(X_t,Y_t)$ far into the future. For example, what is the probability that each species' density is greater than $\epsilon$ in the long term  i.e. $\P[X_t\ge (\epsilon,\dots,\epsilon)]$ for large $t$? The answer to this question provides information about what happens across many independent realizations of the ecological dynamics. Alternatively, one might be interested about the statistics associated with a single realization of the process i.e. a single time series. For instance, one could ask what fraction of the time was each species' density greater than $\epsilon$?

For both ways of quantifying persistence, we define the \emph{extinction set}
\[
\s_0=\{(x,y)\in \s: \min_i x^i =0\}
\]
which corresponds to the set of population states where at least one of the species or genotypes is absent. We also define, for any $\eta>0$, the \emph{$\eta$-neighborhood of the extinction set} by
\[
\s_\eta=\{(x,y)\in \s: \min_i x^i \le \eta\}
\]
that corresponds to states at which at least one of the species has a density less than $\eta$.

With respect to the first approach to studying asymptotic behavior, the model \eqref{eq:main} is \emph{stochastically persistent in probability}~\citep{chesson-82,chesson-ellner-89,jdea-11}  if for all $\varepsilon>0$ there exists an $\eta>0$ such that
\[
\limsup_{t\to\infty}\P_z[X_t\in \s_\eta]\le \varepsilon \mbox{ whenever } z=(x,y) \in \s\setminus \s_0.
\]
In words, there is an arbitrarily small probability of any of the species being at arbitrarily low densities anytime in the future. Moreover, while these uppers bounds depend on initial conditions in the short-term, they are independent of initial conditions (provided all species are initially present) in the long-term. 

For the second approach to study the asymptotic behavior of $Z_t=(X_t,Y_t)$, we define the \emph{occupation measure} by
\[
\Pi_t = \frac{1}{t} \sum_{s = 1}^t \delta_{Z_s}\]
where $\delta_{Z_s}$ denotes a Dirac measure at $Z_s$ i.e. $\delta_{Z_s}(A)=1$ if $Z_s\in A$ and $0$ otherwise for any (Borel) set $A\subset \s$.  $\Pi_t(A)$ equals the proportion of time the community spends in $A$ up to time $t$. Our model \eqref{eq:main} is \emph{almost surely stochastically persistent}~\citep{jmb-11,jdea-11} if for all $\epsilon>0$ there exists $\eta>0$ such that
\[
\limsup_{t\to\infty}\Pi_t(\s_\eta)\le \epsilon \mbox{ with probability one}
\]
whenever $Z_0=z_0\in \s\setminus \s_0.$ This form of persistence implies that fraction of time spent by the community in $\s_\eta$ goes to zero as $\eta$ goes to zero. As with persistence in probability, these upper bounds are asymptotically independent of the initial condition.

Regarding the definition of stochastic persistence, \citet{chesson-82} wrote:
\begin{quote}
This criterion requires that the probability of observing a population below any given density, should converge to zero with density, uniformly in time. Consequently it places restrictions on the expected frequency of fluctuations to low population levels. Given that fluctuations in the environment will continually perturb population densities, it is to be expected that any nominated population density, no matter how small, will eventually be seen. Indeed this is the usual case in stochastic population models and is not an unreasonable postulate about the real world. Thus a reasonable persistence criterion cannot hope to do better than place restrictions on the frequencies with which such events occur. \end{quote}

\subsection{Invariant measures and realized per-capita growth rates}
To determine whether stochastic persistence in either form occurs, we need to understand what happens to a species when it becomes rare in the community. In particular, does the species tend to increase when rare or decrease when rare? To quantify this tendency, we consider the per-capita growth rate of the species averaged over the fluctuations in the community dynamics $X_t$, the auxiliary variable $Y_t$, and the i.i.d. random variables $\xi_{t+1}$. These averages will be taken over invariant probability measures which  represent the stationary dynamics of the common species. Recall, a Borel probability measure $\mu$ on $\s$ is \emph{invariant} if for all continuous functions $h:\s\to\R$
\[
\int_\s h(z)\mu(dz)= \int_\s \E_z[h(Z_1)]\mu(dz)
\]
where
\[
\E_z[h(Z_1)]=\E[h(Z_1)|Z_0=z].
\]
An invariant probability measure $\mu$ is \emph{ergodic} if it can not be written as a non-trivial convex combination of invariant probability measures. 

Given any species $i$, the set $\s^i=\{z=(x,y)\in \s: x^i>0\}$ is the set of states for which species $i$ has positive density. For an ergodic measure $\mu$, $\mu(\s^i)$ is either $1$ or $0$ for any species $i$ i.e. either supports the presence of species $i$ or not. Hence, for an ergodic measure $\mu$, we define $S(\mu)=\{1\le i\le n: \mu(\s^i)=1\}$ as the \emph{species support of $\mu$}. 

To understand coexistence, we can ask: what happens to a species when it becomes rare in the community? To this end, imagine that some species are infinitesimally rare and the dynamics of the remaining, common species are characterized by an ergodic measure $\mu$. As $\mu$ supports only a subset of species, $S(\mu)$ is a proper subset of $\{1,\dots,n\}$. For one of the infinitesimally rare species $i\notin S(\mu)$, its rate of growth is determined by the linearized model
\[
\widetilde X^i_{t+1}=f^i(Z_t,\xi_{t+1})\widetilde X_t^i \mbox{ with }Z_t=(X_t,Y_t)
\] 
for $\mu$-almost every $Z_0=z_0\in \s_0$. The solution to this stochastic linear difference equation is
\[
\widetilde X^i_t=\prod_{\tau=t-1}^0 f^i(Z_\tau,\xi_{\tau+1}) \widetilde X_0^i.
\]
Taking the $\log$ of this solution, dividing by $t$, and taking the limit (provided it exists), species $i$ tends to increase in density if
\[
\lim_{t\to\infty} \frac{1}{t}\sum_{\tau=0}^{t-1}\log f^i(Z_\tau,\xi_{\tau+1})>0
\]
and tends to decrease if
\[
\lim_{t\to\infty} \frac{1}{t}\sum_{\tau=0}^{t-1}\log f^i(Z_\tau,\xi_{\tau+1})<0.
\]
The following proposition shows that for $\mu$-almost every $Z_0=z_0\in \s$,
\[
r^i(\mu)=\lim_{t\to\infty} \frac{1}{t}\sum_{\tau=0}^{t-1}\log f^i(Z_\tau,\xi_{\tau+1})
\]
where $r^i(\mu)$ is  the \emph{realized per-capita growth rate} of species $i$:
\begin{equation}\label{eq:per-capita}
r^i(\mu)=\int_\s \E[\log f^i(z,\xi_t) ] \mu(dz).
\end{equation}
For any infinitesimally rare species $i\notin S(\mu)$, $r^i(\mu)$ determines how quickly the species tends to increase when introduced at low densities. Alternatively,  the realized per-capita growth rate $r^i(\mu)$ equals zero for any species $i$ supported by $\mu$. Intuitively, if the species' density is neither asymptotically growing or declining, then the species' realized per-capita growth must be zero. A proof of the proposition is given in section~\ref{proof:invasion}.

\begin{prop}
\label{prop:invasion}  Let $\mu$ be an invariant  probability measure and $i \in \{1, \ldots, n \}.$ Then there exists a bounded Borel map $\widehat{r}^i:\s\to \R$ such that:
\begin{enumerate}
\item[(i)] With probability one and for $\mu$-almost every $z\in \s$
\begin{equation}\label{eq:r}
\lim_{t\to\infty} \frac{1}{t} \sum_{\tau=0}^{t-1} \log f^i(Z_\tau,\xi_{\tau+1}) =\hat{ r}^i(z)\mbox{ when }Z_0=z;
\end{equation}
and
\begin{equation}\label{eq:r2}
\lim_{t\to\infty} \frac{1}{t} \sum_{\tau=0}^{t-1} \E_z\left[ \log f^i(Z_\tau,\xi_{\tau+1})\right] =\hat{ r}^i(z)\mbox{ when }Z_0=z;
\end{equation}
\item[(ii)] $\int_{\s} \widehat{r}^i(z) \mu(dz) = r^i(\mu);$
\item[(iii)] if $\mu$ is ergodic, then $\widehat{r}^i(z) = r^i(\mu)$ $\mu$-almost surely and $r^i(\mu)=0$ for all $i\in S(\mu).$
\end{enumerate}
\end{prop}

\subsection{Realized per-capita community growth rates and stochastic persistence}\label{subsec:cgr}

Following the approach introduced by Josef Hofbauer~\citep{hofbauer-81,hofbauer-sigmund-98}, our criterion for stochastic persistence is 
\begin{equation}\label{eq:criterion}
\mbox{there exist positive }p^1,\dots,p^n \mbox{ s.t. } \sum_i p^i r^i(\mu)>0\mbox{ for all ergodic $\mu$ with $\mu(\s_0)=1$}.
\end{equation}
Criterion \eqref{eq:criterion} requires that there is some weighting of the species such that the weighted average of the realized per-capita growth is positive for all ergodic measures supporting a subset of the community. In particular, for each ergodic measure $\mu$, \eqref{eq:criterion} requires there is at least one species such that one of the missing species can increase when rare i.e. $r^i(\mu)>0$ for some $i\notin S(\mu)$. In the case of two species models, positive realized per-capita growth rates for at least one missing species is sufficient to ensure that there are positive weights $p^1,p^2$ satisfying \eqref{eq:criterion}. However, for more species, the criterion is more subtle as we illustrate in the evolutionary game of rock-paper-scissors (see section~\ref{sec:eg}).

To see why criterion \eqref{eq:criterion} is useful, we follow the approach developed in \citep{benaim-18,benaim-lobry-16}. To describe this approach, we observe that assumptions \textbf{A1}-\textbf{A2} imply that \eqref{eq:main} is (weak) Feller~\citep{meyn-tweedie-09}: the transition operator $P$ takes continuous functions $h:\s\to\R$ to continuous functions $Ph$ defined by
\[
(Ph)(z)=\E_z[h(Z_1)].
\]
For any $t\ge 0$, let $P_t$ denote the  $t$-time step transition operator of our Markov model \eqref{eq:main} defined by
\[
(P_th)(z)=\E_z[h(Z_t)] \mbox{ for all }z\in \s.
\]
Now, if we define
\[
V(z)=-\sum_{i=1}^n p^i \log(x^i)\mbox{ for }z=(x,y)\in \s,
\]
then 
\[
P_tV(z)-V(z)=-\sum_{i=1}^n p^i \left(\sum_{\tau=0}^{t-1}\E_z[\log f^i (Z_\tau, \xi_{\tau+1})]\right).
\]
Proposition~\ref{prop:invasion} and criterion \eqref{eq:criterion} imply that there is a $T\ge 1$ and $\alpha>0$ such that (see proof in section~\ref{proof:persist})
\begin{equation}\label{eq:key1}
P_TV(z)-V(z)\le -\alpha \mbox{ for all }z\in \s_0.
\end{equation}
For this choice of $T\ge 1$, we prove the following key result in section~\ref{proof:key}.

\begin{prop}\label{prop:key} Assume \eqref{eq:criterion} holds. 
For $\theta>0$ sufficiently small, there exist $\rho \in (0,1)$ and $\beta>0$ such that
\[
(P_T V_\theta)(z) \le \rho V_\theta(z) + \beta  \mbox{ for all }z\in \s\setminus\s_0 \mbox{ where }V_\theta(z)=\exp(\theta V(z)).
\]
\end{prop}

Using Proposition~\ref{prop:key}, criterion \eqref{eq:criterion} implies stochastic persistence in probability. To understand why, note that the $t$-step transition operator $P_t$ maps Borel probability measures $\mu$ to Borel probability measures $\mu P_t$ by duality i.e. $\mu P_t$ is the probability measure defined by
\[
\int_\s h(z)\,(\mu P_t)(dz):=\int_\s (P_th)(z) \mu(dz) \mbox{ for all continuous }h:\s\to\R.
\]
Namely, if $Z_0=z$ is drawn randomly according to $\mu(dz)$, then $\mu P_t$ is the law of $Z_t$. In particular, if $\mu=\delta_{z_0}$ for some $z_0\in \s$, then $\delta_{z_0}P_t$ is the law of $Z_t$ given $Z_0=z.$

Now, for simplicity, assume $T=1$  (see section~\ref{proof:persist} for the general case $T\ge 1$). Given any $z_0\in \s\setminus\s_0$, any integer $t$ and $\eta>0$, Proposition~\ref{prop:key} implies
\begin{equation}~\label{ee}
\P_{z_0}[Z_t\in \s_\eta]\min_{z\in \s_\eta\setminus\s_0}V_\theta(z)\le \int_\s V_\theta(z)d(\delta_{z_0}P_t) \le (\rho)^t V_\theta(z_0)+\frac{\beta}{1-\rho}.
\end{equation}
For $\eta\le 1$ and $z=(x,y)\in \s_\eta\setminus \s_0$,
\begin{equation}\label{eq:lb}
V_\theta(z)=\prod_{i=1}^n (x^i)^{-\theta\,p^i}\ge a_0(\eta)^{-b} \mbox{ where }a_0=\min_{z=(x,y)\in\s\setminus \s_0}\|x\|^{b(1-n)},b=\theta\max_i p^i.
\end{equation}
Thus, for $\eta\le 1$, inequalities \eqref{ee}--\eqref{eq:lb} and $\rho<1$ imply
\[
\limsup_{t\to\infty}\P_{z_0}[Z_t\in \s_\eta]\le a(\eta)^b \mbox{ where }a=\frac{\beta}{a_0(1-\rho)}
\]
which implies persistence in probability.

To show that criterion \eqref{eq:criterion} implies almost sure stochastic persistence requires two ingredients. First, using an extension of an argument presented in \citet{jmb-11}, we show in section~\ref{proof:persist} that if $Z_0\in \s\setminus \s_0$, then all weak* limit points of $\Pi_t$ are, with probability one, invariant probability measures supported by $\s\setminus \s_0$. Second, we show that the weight placed by positive invariant probability measures  near the boundary can be controlled uniformly. Specifically, if $\mu$ is an invariant probability measure with $\mu(\s_0)=0$ and $\int V_\theta(z)\mu(dz)$ is finite (the proof of the general case without this integrability assumption appears in section~\ref{proof:persist}), then Proposition~\ref{prop:key} implies for any $k\ge 1$
\[
\int V_\theta(z) \mu(dz)=\int (P_{kT}V_\theta)(z) \mu(dz) \le (\rho)^k \int V_\theta(z) \mu(dz) +\frac{\beta}{1-\rho}.
\]
As $\rho<1$, taking the limit as $k\to\infty$ yields
\[
\int V_\theta(z) \mu(dz) \le \frac{\beta}{1-\rho}.
\]
Thus, for any $\eta>0$,
\[
\mu(\s_\eta)\min_{z\in \s_\eta}V_\theta(z)\le\int V_\theta(z) \mu(dz) \le \frac{\beta}{1-\rho}.
\]
Inequality \eqref{eq:lb} implies that
\[
\mu(\s_\eta)\le a (\eta)^b \mbox{ for all }\eta\le 1  \mbox{ where }a=\frac{\beta}{a_0(1-\rho)}.
\]

\begin{thm}\label{thm:persist} If criterion \eqref{eq:criterion} holds,
then \eqref{eq:main} is almost-surely stochastically persistent and stochastically persistent in probability. In particular, there exists  $a,b>0$ such that for all $\eta\le 1$ and $z\in \s\setminus \s_0$
\[
\mbox{\rm(persistence in probability) }\limsup_{t\to\infty}\P_z\left[Z_t\in \s_\eta \right]\le a (\eta)^b 
\]
and
\[
\mbox{\rm(almost-sure persistence) }\limsup_{t\to\infty}\Pi_t\left[ \s_\varepsilon \right]\le a (\eta)^b \mbox{ for }Z_0=z.
\mbox{ almost surely.}\]
\end{thm}

\subsection{Escaping extinction risk.} When stochastic persistence occurs, it is natural to ask: ``if one or more species are at very low densities, how long before they reach higher densities?'' For $\eta>0$ sufficiently small, the following proposition provides upper bounds on the time for all species densities to exceed $\eta$. 

\begin{prop} Assume that Proposition~\ref{prop:key} holds. There exists $\eta>0$ such that if $Z_0=(x,y)\in \s_\eta\setminus \s_0$ and \[\tau=\inf_{k\ge 1}\{k:Z_{kT}\in \s\setminus \s_\eta\},\] then \begin{equation}\label{eq:escape}\P[\tau>k]\le \prod_{i=1}^n (x^i)^{-\theta p^i}(\rho)^k.\end{equation}
\end{prop}

\emph{\textbf{Biological implication.} Inequality~\eqref{eq:escape} implies that the probability that one or more species remain at low densities  decreases exponentially over time. Furthermore, the larger the minimal community realized per-capita growth rate ($\alpha$ in \eqref{eq:key1}), the greater this rate of exponential decline.}

\begin{proof}
From the proof of Proposition~\ref{prop:key}, we can choose $\eta>0$ sufficiently small so that $P_TV_\theta(z)\le \rho\, V_\theta(z)$ for all $z\in \s_\eta \setminus \s_0.$ \citet[Proposition 8.2]{benaim-18} implies that $\E_z[\rho^{-\tau}]\le V_\theta(z)=\prod_i (x^i)^{-\theta p^i}$ for $z=(x,y)\in \s_\eta\setminus \s_0$. By the Markov inequality, 
\[
\P_z[\tau\le k]=\P_z[\rho^{-\tau}\ge \rho^{-k}]\le \E_z[\rho^{-\tau}]\rho^k \le \prod_i (x^i)^{-\theta p^i} \rho ^k. 
\]  
\end{proof}

\section{Stochastic extinction}\label{sec:extinction}

To truly understand the conditions for stochastic persistence, we also need to understand when the solutions of our model~\eqref{eq:main} converge to the extinction set $\s_0$ with positive probability. To this end, we prove two theorems about extinction. The first theorem is a partial converse to Theorem~\ref{thm:persist}, and shows that when the inequality \eqref{eq:criterion} is reversed for all ergodic measures supported by $\s_0$, population trajectories starting near $\s_0$ are likely to asymptotically approach $\s_0$. To measure the rate of approach to the extinction set $\s_0$, define
\[
\rm{dist}(z,\s_0)=\min_{z'\in \s_0}\|z-z'\|
\]
where $\|\cdot\|$ denotes the standard Euclidean norm. Our second main extinction theorem considers the case when there is a subset of persisting species that can not be invaded by the other species i.e. $r^i(\mu)<0$ for all species not in the community. In this case, we show that when the missing species are sufficiently rare, they are highly likely to go extinct. For both extinction results, we also introduce an accessibility criterion that, when met, implies that all initial conditions lead to asymptotic extinction at an exponential rate with probability one.

\begin{thm}\label{thm:exclusion1}  If there exist positive weights $p^1,p^2,\dots,p^n$ such that
\begin{equation}\label{eq:exclusion1}
\sum_{i=1}^n p^i r^i (\mu) <0
\end{equation}
for all ergodic $\mu$  supported by $\s_0,$ then there exists $a,b>0$ such that
\[
\P_z\left[\limsup_{t\to\infty}\frac{1}{t}
\log\rm{dist}(Z_t,\s_0)<0\right]\ge 1 - a\,\rm{dist}(z,\s_0)^b.
\]
\end{thm}
The proof of Theorem~\ref{thm:exclusion1} is given in section~\ref{proof:exclusion1} whose strategy is based on the proof of Theorem 3.3 of \citet{benaim-lobry-16} for random switching between two dimensional Lotka-Volterra competition equations. A key ingredient of the proof is to introduce the function $V(z)=\sum_{i=1}^np^i \log x^i$--the negative of the $V$ function used in the proof of Theorem~\ref{thm:persist}. The analog of Proposition~\ref{prop:key} holds for this choice of $V$. However, in this case, the function $V_\theta(z)=\prod_{i=1}^n (x^i)^{p^i}$ tends to decrease after $T$ iterates near the extinction set $\s_0$.

To get asymptotic extinction with probability one for all initial conditions, we say \emph{$\s_0$ is accessible} if for all $\eta>0$, there exists $\gamma>0$ such that
\[
\P_z[ X_t\in \s_\eta \mbox{ for some }t\ge 1]\ge \gamma\mbox{ for all }z\in \s.
\]
The following corollary implies that accessibility plus the conditions of Theorem~\ref{thm:exclusion1} imply almost sure extinction. We prove this corollary in section~\ref{proof:exclusion1}.

\begin{cor}\label{cor1} Assume the assumptions of Theorem~\ref{thm:exclusion1} hold and $\s_0$ is accessible. Then
\[
\P_z\left[\limsup_{t\to\infty}\frac{1}{t}\log\rm{dist}(Z_t,\s_0)<0
\right]=1 \mbox{ for all }z\in \s.
\]
\end{cor}

Our second main extinction result considers the case when there is a set of species that can not be invaded by another other species. Recall that for an ergodic probability measure $\mu$, $S(\mu)\subset\{1,2,\dots,n\}$ is the set of species supported by $\mu.$

\begin{thm}\label{thm:exclusion2} Let $\{1,\dots,n\}=I\cup J$ where $I\cap J=\emptyset$ and $J\neq \emptyset$. Assume
\begin{enumerate}
\item[(i)]\eqref{eq:main} restricted $\s^I:=\{(x,y)|x^j=0$ whenever $j\in J\}$ satisfies that there exists $p^i>0$ for $i\in I$ and $\sum_{i\in I}p^ir^i(\mu)>0$ for ergodic $\mu$ with $\mu(\s^I_0)=1$ where $\s^I_0:=\{z=(x,y)\in \s^I: \prod_{i\in I}x^i=0\}$, and
\item[(ii)] $r^j(\mu)<0$ for any $j\in J$ and ergodic $\mu$ with $S(\mu)=I.$
\end{enumerate}
Then there exist $a,b,c>0$ such that
\begin{equation}\label{eq:exclusion2}
\P_z\left[\limsup_{t\to\infty} \frac{1}{t} \log(X_t^j)<0 \mbox{ for all }j\in J\right]\ge 1 - c \frac{\left(\max_{j\in J} x^j\right)^a}{\left(\min_{i\in I} x^i\right)^b}
%1-a \left(\frac{\max_{j\in J}x^j}{\prod_{i\in I}x^i}  \right)^b
\end{equation}
whenever $z=(x,y)$ satisfies $\prod_{i\in I}x^i>0.$
\end{thm}

The lower bound in equation~\eqref{eq:exclusion2} implies that for a given set of positive densities for the species in subcommunity $I$ (i.e. $x^i>0$ for $i\in I$), then the probability of asymptotic extinction of an unsupported species increases to one as their densities get sufficiently small. The proof of Theorem~\ref{thm:exclusion2} follows the proof strategy used for unstructured stochastic differential equations by \citet{hening-nguyen-18}. For this proof, the key ingredient is the function $V(z)=-\sum_{i\in I}p^i \log x^i + \delta \max_{j\in J} \log x^j$ for a sufficiently small $\delta>0.$ The corresponding $V_\theta(z)=\exp(\theta V(z))$ for $\theta>0$ sufficiently small tends to be increasing at low densities of  species from $I$ and decreasing at low densities of species from $J.$

To get extinction with probability one, we need an accessibility assumption. To allow for  multiple, uninvasible subcommunities, we consider a finite collection of proper subsets of species, $I_1,\dots,I_m\subset \{1,\dots,n\}$. For each $1\le i\le m$ and $\eta>0$, define $\s^{I_i}_\eta=\{(x,y)\in\s:x^j\le \eta \mbox{ for all }j\in I_i\}$. We say \emph{$\cup_i\s_{I_i}$ is accessible} if for all $\eta>0$, there exists $\gamma>0$ such that
\[
\P_z[ Z_t\in\cup_i \s^{I_i}_\eta \mbox{ for some }t\ge 1]\ge \gamma
\]
whenever $z=(x,y)$ satisfies $\prod_{i}x^i>0.$

\begin{cor}\label{cor2} Assume the assumptions of Theorem~\ref{thm:exclusion2} hold for each $I_i$, $J_i=\{1,\dots,n\}\setminus I_i$ and $\cup_i\s^{I_i}$ is accessible. Then
\[
\P_z\left[\limsup_{t\to\infty} \frac{1}{t}\log\rm{dist}(Z_t,\s_0)<0\right]=1 \mbox{ for all }z\in\s.
\]
\end{cor}

The proof of this corollary uses the same strategy as the proof of Corollary~\ref{cor1} in section~\ref{proof:exclusion1}. The details are left to the reader.

\section{Applications}\label{sec:apps}

To illustrate the utility of our theorems, we apply them to four types of models. Our first set of models considers a stochastic version of two and three strategy evolutionary games. It highlights how Theorem~\ref{thm:persist} and Theorems~\ref{thm:exclusion1}--\ref{thm:exclusion2} can be used to fully characterize the possible dynamics of coexistence and extinction in these evolutionary games. The second set of models are stochastic counterparts to the Lotka-Volterra differences equations~\citep{hofbauer-etal-87}. These models exhibit an averaging property that allows one to readily compute the realized per-capita growth rates $r^i(\mu)$ with respect to any ergodic measure. We illustrate the use of this averaging property with a predator-prey model. Our last two models are single species models with different types of auxiliary variables. The first of these models considers a species evolving in a fluctuating environment. This example illustrates how, simultaneously, an evolving trait can be incorporated as an internal variable and how auto-correlated fluctuations can incorporated as an external variable.  The second of these two  models explores the dynamics of  disease spread in a  spatially structured population. This examples  illustrates how discrete population structure can be represented as internal variables. 

\subsection{Stochastic evolutionary games}\label{sec:eg} Evolutionary game theory is an important framework for studying frequency-dependent population dynamics~\citep{maynardsmith-82,hofbauer-sigmund-98,hofbauer-sigmund-03}. Remarkably, the three basic games (the Prisoner's Dilemma, the Hawk-Dove game, and the Rock-Paper-Scissor game) have provided fundamental insights about the evolution of cooperation \citep{axelrod-84,axelrod-hamilton-81}, animal contests \citep{maynardsmith-82,maynardsmith-price-73}, and Red Queen  dynamics \citep{sinervo-lively-96,kerr-etal-02,kirkup-riley-04,nahum-etal-11}. An important, yet often under-appreciated, consideration in these games is the effect of environmental stochasticity on the maintenance or loss of  polymorphisms i.e. coexistence of multiple strategies. Here, we show how our results characterize (generically) persistence in probability and extinction for stochastic versions of these classic games in correlated, as well as, uncorrelated environments.

 For the stochastic evolutionary games, there are $n$ competing asexual genotypes or strategies. The frequency of strategy or genotype $i$ is $X^i$. Individuals interact randomly resulting in fitness payoffs that modify the basal fitness of each strategy. Let $b^i(Y)$ be the basal fitness (possibly zero) of strategy $i$ and $A^{ij}(Y)$ be the fitness payoff (possibly negative) to an individual of strategy $i$ following an interaction with an individual of strategy $j$. Here $Y$ is an environmental variable that determines the fitness payoffs in any time step. If $Y$ is given by a first-order, multivariate auto-regressive process, then the evolutionary game dynamics are
\begin{equation}\label{eq:lottery}
\begin{aligned}
X^i_{t+1}=&X_t^i\frac{\sum_jA^{ij}(Y_{t})X_t^j+b^i(Y_{t})}{\sum_{j\ell}A^{j\ell}(Y_{t})X^j_tX^\ell_t+\sum_j b^j(Y_{t})X^j_t}\\
Y_{t+1}=&C Y_t +\xi_{t+1}
\end{aligned}
\end{equation}
where $C$ is a matrix with a spectral radius strictly less than one, and $\xi_1,\xi_2,\dots$ is a sequence of i.i.d random variables taking value in a compact subset of $\R^n.$ Under these assumptions on the auxiliary variable dynamics, they enter and remain in a compact set $K\subset \R^k$ with probability one. Hence, the state space is given by $\s=\Delta\times K$ where $\Delta=\{(x^1,\dots,x^n)\in [0,1]: \sum_i x^i=1\}$ is the probability simplex.

We illustrate our results with two important evolutionary games: the hawk-dove game and the rock-paper-scissors game.

\subsubsection{Hawk-dove game.} For $n=2$ strategy games, $\s_0=\{(1,0),(0,1)\}\times K$ which only supports two ergodic measures $\mu^1=\delta_{(1,0)}\otimes \nu$ and $\mu^2=\delta_{(0,1)}\otimes \nu$ where $\otimes$ denotes the product of two measures, and $\nu$ is the law of the stationary distribution $\widehat Y$ of the auto-regressive process $Y_t$. It follows that \[
r^j(\mu^i)=\E\left[\log\frac{A^{ji}(\widehat Y)+b^j(\widehat Y)}{A^{ii}(\widehat Y)+b^i(\widehat Y)}\right].
\]
If both strategies can invade each other ($r^j(\mu^i)>0$ for $i\neq j$), Theorem~\ref{thm:persist} implies stochastic persistence. Alternatively, if $r^j(\mu^i)<0$ for some $i\neq j$, then Theorem~\ref{thm:exclusion2} implies strategy $j$ is excluded with high probability when its initial frequency is sufficiently low.

For example, consider the hawk-dove game in which hawks (strategy $1$) and doves (strategy $2$) engage in pairwise contest over a resource worth $V(Y)$ offspring. Interacting doves, split the resource, while hawks get all of the resource when encountering a dove. Encounters between hawks result in one getting all of the resource and the other paying a cost $C$. Assume that all individuals have a base pay-off of $b(Y)\equiv 1$ and that $V(Y)-C>-2$ to ensure that payoffs are always positive. Then
\[
A(Y)=1+\begin{pmatrix} \frac{V(Y)-C}{2}& V(Y)\\
0& V(Y)/2\end{pmatrix}.
\]
Hawks always invade a population of doves as \[r^1(\mu^2)=\E\left[\log\left(\frac{1 + V(\hat Y)}{1+V(\hat Y)/2)}\right)\right]>0.\] Doves can invade a population of hawks only if \[
r^2(\mu^1)=\E\left[\log \frac{1}{(V(\hat Y)-C)/2+1}\right]>0.
\]
Provided $\rm{Var}[V(\widehat Y)]>0$, Jensen's inequality implies
\[
\E\left[\log \frac{1}{(V(\hat Y)-C)/2+1}\right]
>\log \frac{1}{(\E[V(\hat Y)]-C)/2+1}.
\]\emph{\textbf{Biological implication.} Environmental fluctuations in the payoff $V(Y)$ can facilitate the coexistence of hawk and dove strategies. Importantly, unlike the deterministic hawk-dove game, hawks and doves can coexist even if the mean reward for winning a contest is greater than the cost paid by a hawk losing a battle to another hawk .}

\subsubsection{Rock-paper-scissor dynamics} To illustrate the use of Theorem~\ref{thm:exclusion1}, we consider the rock-paper-scissor game. For mathematical convenience, we number the strategies $0$ (rock), $1$ (paper), and $2$ (scissors). When strategy $i$ interacts with strategy $i+1$ (mod $3$), it receives a negative payoff $-L_t^i<0$ while the other strategy receives a positive payoff $W_t^{i+1}$. Let $Y_t=(W^0_t,W_t^1,W_t^2,L_t^0,L_t^1,L_t^2)$. Then the payoff matrix equals
\[
A(Y_t)=\begin{pmatrix}
0&-L_t^1&W_t^1\\
W_t^2&0&-L_t^2\\
-L_t^3&W_t^3&0\\
\end{pmatrix}.
\]
Assume a base payoff $b(Y)\equiv 1$ of one and $L_t^i<1$. Furthermore, for simplicity, assume $L^0_t,L^1_t,L^2_t$ are identically distributed and $W_t^0,W_t^1,W_t^2$ are identically distributed. 

\begin{figure}
\includegraphics[width=0.6\textwidth]{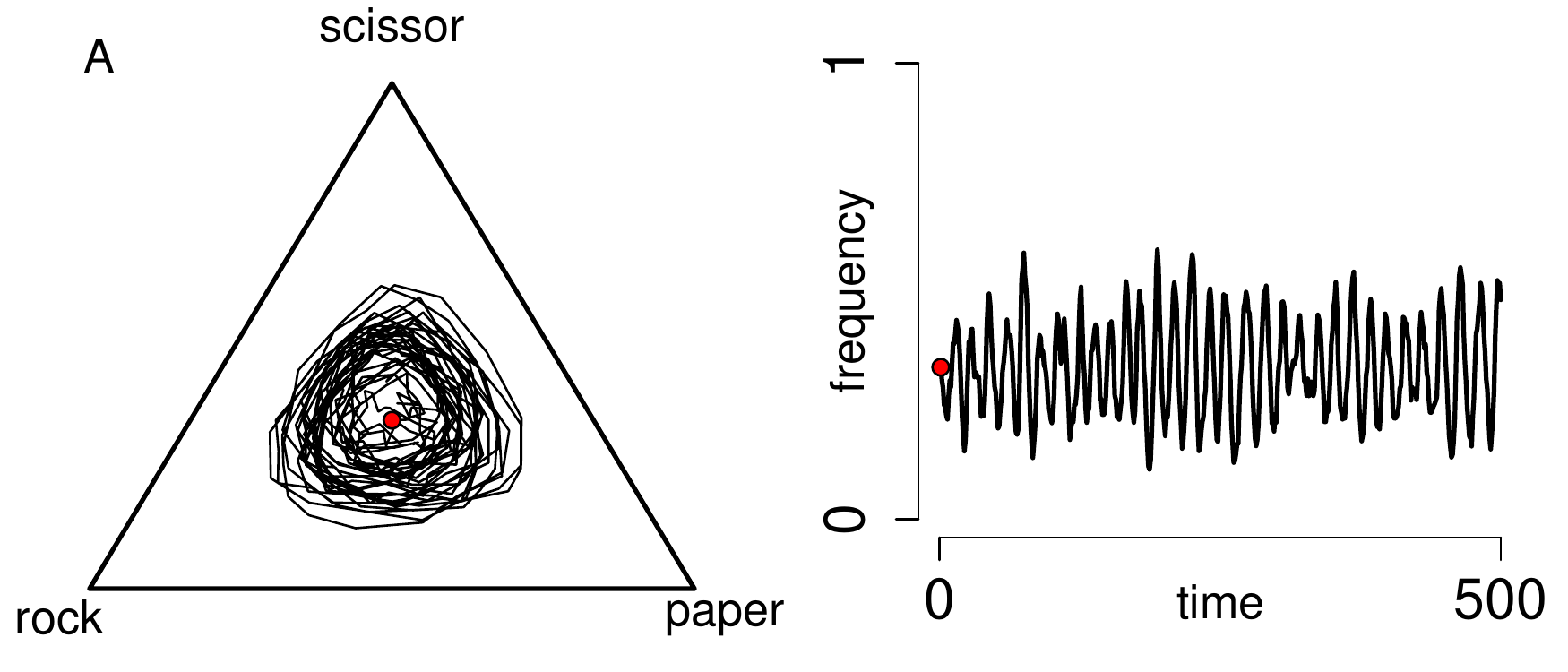}\\
\includegraphics[width=0.6\textwidth]{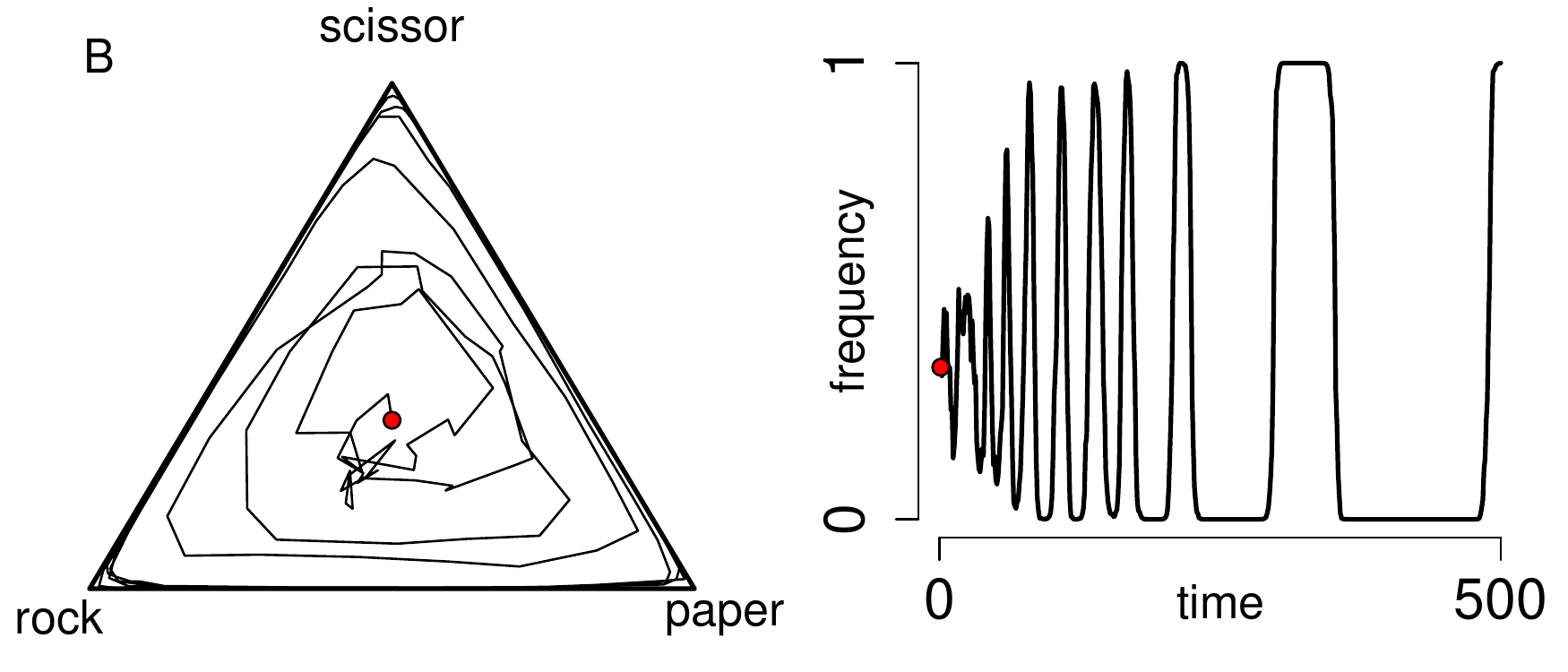}
\caption{Noise induced exclusion in the rock-paper-scissor dynamic. In both panels,  payoffs fluctuate around around the mean values $\E[\widehat W_t^i]=2.3$ and $\E[\widehat L_t^i]=0.5$. In (A), payoffs are equally likely to be $\pm 40\%$ of the mean payoff. In (B), payoffs are equally likely to be $\pm 90\%$ of the mean payoff. }\label{fig:rps}
\end{figure}

As $1-L_t^i<1+W_t^{i+1}$ implies that for any pair of strategies $i$ and $i+1$ (mod $3$), strategy $i$ gets excluded by strategy $i+1$ (mod $3$) i.e. $\lim_{t\to\infty} X^{i+1}_t=1$ with probability one whenever $X_0^i+X_0^{i+1}=1$ and $X_0^{i+1}>0$. Hence, the only invariant measures on the boundary $\s_0$ are given by the pure strategy invariant measures $\delta_i\otimes \nu$ where $\delta_i$ is a Dirac measure only supporting strategy $i$ and $\nu$ is the law of the stationary distribution $\widehat Y$ of the autoregressive process. The realized per-capita growth rates at these invariant measures are given by
\[
r^{i-1 \mbox{\tiny (mod 3)}}(\delta_i\otimes \nu)=\E[\log (1-\widehat L^1)],r^i(\delta_i\otimes \nu)=0,\mbox{ and }r^{i+1 \mbox{\tiny (mod 3)}}(\delta_i\otimes \nu)=\E[\log(1+\widehat W^1)].
\]
A simple inequality calculation reveals that there are weights $p^i$ such that $\sum_i p^i r^i(\delta_j\otimes \nu)>0$ for all $j$
if and only if
\begin{equation}\label{eq:rps1}
\E[\log(1+W_t^1)]+\E[\log(1-L_t^1)]>0.
\end{equation}
When this inequality is meet, Theorem~\ref{thm:persist} implies all three strategies stochastically persist. Alternatively, there are weights $p^i$ such that $\sum_i p^i r^i(\delta_j\otimes \nu)<0$ for all $j$
if and only if
\begin{equation}\label{eq:rps2}
\E[\log(1+W_t^1)]+\E[\log(1-L_t^1)]<0.
\end{equation}
When this inequality is meet, Theorem~\ref{thm:exclusion1} implies that $Z_t$ asymptotically approaches $\s_0$ with high probabilities whenever one of the strategies is at low frequency. Hence, the conditions for robust stochastic persistence~\eqref{eq:rps1} and exclusion~\eqref{eq:rps2} are sharp except for the degenerate case in which the left hand side of \eqref{eq:rps1} equals zero.

So what effect does environmental stochasticity have on maintaining a polymorphism? When $\rm{Var}[W^1_t]+\rm{Var}[L^1_t]>0$, Jensen's inequality implies
\[
\E[\log(1+W_t^1)]+\E[\log(1-L_t^1)]<\log(1+\E[W_t^1])+\log(1-\E[L_t^1]).
\]
\emph{\textbf{Biological implication.} Environmental fluctuations in the payoffs inhibit coexistence of the rock, paper, scissors strategies. In particular, unlike the deterministic games, even if the average winning payoff is greater than the average losing payoff, the three strategies may not coexist} (Fig.~\ref{fig:rps}).

\subsection{Stochastic Lotka-Volterra difference equations}

\citet{hofbauer-etal-87} introduced a discrete-time counterpart to the Lotka-Volterra differential equations. These equations consider $n$ interacting species whose dynamics in the absence of environmental stochasticity are
\begin{equation}\label{eq:LVdet}
X^i_{t+1}=X^i_{t+1} \exp\left(\sum_j a^{ij} X^i_t +b^i\right) \quad i=1,2,\dots,n
\end{equation}
where $a^{ij}$ describes the per-capita effect of species $j$ on species $i$'s per-capita growth rate and $b^i$ corresponds to the intrinsic rates of growth of species $i$. \citet{hofbauer-etal-87} showed that these difference equation share many mathematical properties of the Lotka-Volterra differential equations. Most importantly, as we describe below for their stochastic counterpart, they exhibit an averaging property that allows one to verify permanence or exclusion by solving linear inequalities. 

Here, we consider a stochastic counterpart to these equations by allowing the interaction strengths $a^{ij}$ and intrinsic per-capita growth rates $b^i$ to vary stochastically. Specifically, we assume that for each $i,j$ the sequence $a^{ij}_1,a^{ij}_2,a^{ij}_3,\dots$ are i.i.d. and for each $i$ the sequence $b_1^i,b_2^i,b_3^i,\dots$ are i.i.d. Then the dynamics become
\begin{equation}\label{LV}
X^i_{t+1}=X^i_{t+1} \exp\left( \sum_{j}a_{t+1}^{ij}X_t^j +b_{t+1}^i\right)\quad i=1,2,\dots,n.
\end{equation}
This system of stochastic difference equations generalizes the one introduced by \citet{jmb-11} who only considered fluctuations in the $b_t^i$. Then $\xi_t=(a^{ij}_t,b^i_t)_{1\le i,j\le j}$ for these models. 

To ensure the dynamics of \eqref{LV} remain in a compact set, we follow \citet{hofbauer-etal-87} and define the interaction strengths $a_t^{ij}$ to be \emph{hierarchically ordered} if there exists a reordering of the indices such that $a_t^{ii}<0$ for all $i$ and $t\ge 0$, and $a^{ij}_t\le 0$ whenever $i\le j$ and $t\ge 0$. While this assumption excludes 	 mutualistic interactions, it allows for many types of predator-prey or competitive interaction. The following lemma shows that hierarchically ordered systems with bounded interaction strengths and intrinsic per-capita growth rates exhibit bounded dynamics. 

\begin{lem}\label{lem:bounded} If stochastic Lotka-Volterra model \eqref{LV} is hierarchically ordered and there exist $\alpha,\beta>0$ such that $a^{ii}_t\le -\alpha$ and $|a^{ij}_t|,|b_t^i| \le \beta$ for all $i,j$ and $t\ge 0$, then there exists $K>0$ such that $X_t\in [0,K]^n$ for $t\ge n+1$. \end{lem}

\begin{proof} Following \citet{hofbauer-etal-87} observe that
\[
X_{t+1}^1 \le X_t^1 \exp(-\alpha X_t^1+\beta)
\mbox{ for all }t \]
as $a^{1j}_t\le 0$ for all $j\ge 2$ and $t\ge 0$. Hence, $X_t^1\in [0,K_1]$ for  $t\ge 2$ where $K_1=\exp(\beta-1)/\alpha$.

Assume that there exist $K_1,\dots, K_{j-1}>0$ such that $X_t^i \in [0,K_i]$ for $i\le j-1$ and $t\ge i+1$. We will show that there exists $K_{j}$ such that $X_t^{j}\in [0,K_{j}]$ for $t\ge j+1$. Indeed, by the hierarchically ordered assumption and our inductive assumption,
\[
X_{t+1}^{j}\le X_t^{j}\exp(-\alpha X^{j}_t+\beta+\sum_{i<j}\beta K_i)
\]
for $t\ge j$. Hence, $X_t^j \le K_j$ for $t\ge j+1$ where
 $K_j=\exp(\beta+\beta\sum_{i<j}  K_i-1)/\alpha$. Defining $K=\max K_j$ completes the proof.
\end{proof}

The following lemma shows that verifying the conditions for Theorems~\ref{thm:persist},\ref{thm:exclusion1}, and \ref{thm:exclusion2} reduces to a linear algebra problem. In particular, this lemma implies that many of the permanence and exclusion criteria developed by \citet{hofbauer-etal-87,hofbauer-sigmund-98} for hierarchal systems extends to the stochastic Lotka-Volterra difference equations by replacing the deterministic terms $a^{ij}$ and $b^i$ in \eqref{eq:LVdet} with the expectations $\E[a^{ij}_t]$ and $\E[b^i_t]$ in \eqref{LV}. 

\begin{lem}  Let $\mu$ be an ergodic measure for \eqref{LV}.   If there exists a unique non-negative solution $\hat x$ to
\[
\sum_j   \E[{a}^{ij}_t]\hat x^j+\E[b^i_t]=0 \mbox{ for }i\in S(\mu)\mbox{ and }
\hat x^i=0 \mbox{ for } i\notin S(\mu)
\]
 then
\[
r^i(\mu)=\left\{\begin{array}{cc} 0& \mbox{ if }i\in S(\mu)\\ \sum_j\E[a^{ij}_t]  \hat x^j +\E[b_t^i] & \mbox{otherwise.}
\end{array}\right.
\]
\end{lem}

\begin{proof}
Let $\mu$ be an ergodic measure. Assertion $(iii)$ of Proposition~\ref{prop:invasion}  implies that
\[
0=r^i(\mu)= \sum_j \E[a^{ij}_t] \int x^j\,\mu(dx)+\E[b^i_t]
\]
for all $i\in S(\mu)$. Since we have assumed there is a unique non-negative solution $\hat x^i$ to this system of linear equations with the additional condition $\hat x^i=0$ for $i\notin S(\mu)$, it follows that $\int x^i \mu(dx)=\hat x^i$ for all $i$. Thus, for any $i\notin S(\mu)$,
\[
r^i(\mu)=\sum_j \E[a^{ij}_t] \int x^j\,\mu(dx)+\E[b^i_t]=\sum_j \E[a^{ij}_t] \hat x^j+\E[b^i_t]
\]
as claimed.
 \end{proof}

To illustrate the use of the lemma, lets consider a predator-prey system with a fluctuating carrying capacity $K_t$ of the prey. In this system, $X_t^1$ and $X_t^2$ correspond to the prey and predator densities, respectively. Let $\rho$ be the intrinsic rate of growth of the prey, $a$ the attack rate of the predator on the prey, $b$ the conversion of prey to predator, $c$ the predator's intraspecific competition term, and $d$ the predator's per-capita density-independent death rate. The resulting hierarchically ordered model is given by 
\begin{equation}\label{eq:pred-prey}
\begin{aligned}
X_{t+1}^1 =& X_t^1 \exp\left(\rho(1-X_t^1/K_{t+1})-a X_t^2\right)\\
X_{t+1}^2=& X_t^2 \exp\left(bX_t^1-cX_t^2-d\right).
\end{aligned}
\end{equation} 
As long as $K_t$ take values in a compact subset of $(0,\infty)$, Lemma~\ref{lem:bounded} implies that $X_t$ enters and remains in a compact subset $\s$ of $\R^2_+$. Let $\delta_{(0,0)}$ be the Dirac measure at the origin. Then $r^1(\delta_{(0,0)})=\rho>0$ and $r^2(\delta_{(0,0)})=-d$. Any other ergodic measure $\mu$ supported by $\s_0$ must be supported by prey axis $(0,\infty)\times \{0\}$. For such an ergodic $\mu$, Proposition~\ref{prop:invasion}(iii) implies that 
\[
0=r^1(\mu)=\rho-\rho\underbrace{\int x^1 \mu(dx^1)}_{=:\hat x^1}\E[1/K_t].
\]
Hence, $\hat x^1=1/\E[1/K_t]$ i.e. the harmonic mean of $K_t$. On the other hand, 
\[
r^2(\mu)=b\hat x^1-d=b/\E[1/K_t]-d.
\] 
If $r^2(\mu)>0$, then Theorem~\ref{thm:persist} with $p^1=2d/\rho$, $p^2=1$ implies \eqref{eq:pred-prey} is persistent. On the other hand, if $r^2(\mu)<0$, then Theorem~\ref{thm:exclusion2} (with the prey subsystem being the persistent subsystem) implies that asymptotic extinction of the predator occurs with high likelihood whenever the predator density is sufficiently low. 

\emph{\textbf{Biological implication.} As the $d/b$ can be interpreted as the predator's break even point (i.e. the prey density at which its per-capita growth rate is zero in the absence on interference), we have shown that predator-prey persistence requires that the harmonic mean of the carrying capacity exceeds the predator's break even density. As the harmonic mean is less than the arithmetic mean, it follows that fluctuations in the prey's carrying capacity inhibit predator-prey coexistence.}

\subsection{Trait evolution in a stationary environment}\label{subsec:trait}

%%% All symbolic calculations in the file EvolutionaryRescue.wxm

Following \citet{lande-shannon-96}, we consider a population with density $X_t$  and a trait (e.g. log body size) $z$ which is normally distributed in the population with mean $\bar z_t$ in year $t$ and constant variance $\sigma^2$. The optimal trait value over the time interval $(t,t+1]$ is $\theta_{t+1}$. The fitness of a individual with trait $z$ over this time step equals
\[
W_{t+1}(z)=\exp(r_{\mbox{\tiny{max}}}-\gamma/2(z-\theta_{t+1})^2-a(\bar z_t) X_t)
\]
where $\gamma$ determines the strength of stabilizing selection and $a(\bar{z})>0$ corresponds to  the strength of intraspecific competition as a function on the mean trait value. Integrating across the normal distribution of the population gives the mean fitness $\overline{W}_{t+1}$ of the population over the time interval $(t,t+1]$
\[
\begin{aligned}
\overline{W}_{t+1}:=&\frac{1}{\sqrt{2\pi\sigma^2}}\int_{-\infty}^{\infty}W_{t+1}(z)\exp(-(z-\bar z_t)/(2\sigma^2))\,dz\\
=&\frac{1}{\sqrt{\gamma\sigma^2+1}}\exp\left(r_{\mbox{\tiny{max}}}-\frac{(\bar{z}_t-\theta_{t+1})^2\gamma}{2\sigma^2\gamma+2}-a(\bar{z}_t)X_t \right).
\end{aligned}
\]
\citet{lande-76} has shown that if $h^2 \in (0,1]$ is the heritability of the trait $z$, then the population and the mean trait dynamics can be approximated by
\begin{equation}\label{ee:one}
\begin{aligned}
X_{t+1}=&X_t \overline{W}_{t+1}\\
\bar z_{t+1}=&\bar z_t + \sigma^2 h^2 \frac{\partial \log \overline{W}_{t+1}}{\partial \bar z_t}\\
=& (1-\alpha)\bar{z}_t+\alpha \theta_{t+1}-\sigma^2h^2a'(\bar{z}_t)X_t\mbox{ where }\alpha=\frac{ h^2\sigma^2\gamma }{\sigma^2\gamma+1}<1.
\end{aligned}
\end{equation}
To finalize the model, we model the optimal trait dynamics as a first-order autoregressive process:
\begin{equation}\label{ee:two}
\theta_{t+1}=\rho \theta_t + \sqrt{1-\rho^2}\xi_{t+1}
\end{equation}
where $\rho\in (-1,1)$ determines the temporal autocorrelation, and $\xi_1,\xi_2,\xi_3,\dots$ is an i.i.d. sequence of random variables with mean $0$, variance $\tau^2$, and compact support. The rescaling of $\xi_{t+1}$ by $\sqrt{1-\rho^2}$ ensures that the variance of the stationary distribution of $\theta_t$ equals $\tau^2$ i.e. $\rho$ only influences the autocorrelation between $\theta_t$ and $\theta_{t+1}$ but not the long-term variance of $\theta_t.$

In the context of our general model formulation~\eqref{eq:main},  $Y=(\bar z, \theta)$ is the auxiliary vector, the fitness $f$ for the species is
\[
f(X,Y,\xi)=\frac{1}{\sqrt{\gamma\sigma^2+1}}\exp\left(r_{\mbox{\tiny{max}}}-\frac{(\bar{z}-\rho\theta-\sqrt{1-\rho^2}\xi)^2\gamma}{2\sigma^2\gamma+2}-a(\bar z)X \right),
\]
and the auxiliary function $G$ is
\[
G(X,Y,\xi)=\left((1-\alpha)\bar z+\alpha(\rho \theta+\sqrt{1-\rho^2}\xi)-\sigma^2h^2 a'(\bar{z})X,\rho\theta+\sqrt{1-\rho^2}\xi \right).
\]

\begin{figure}
\includegraphics[width=0.8\textwidth]{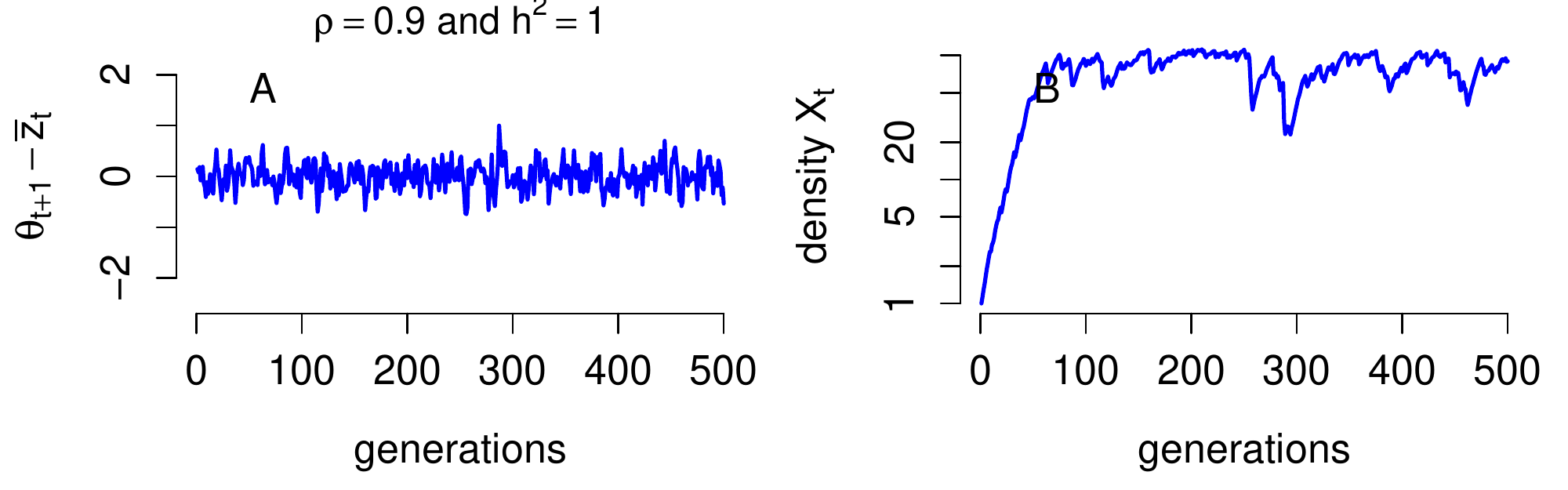}
\includegraphics[width=0.8\textwidth]{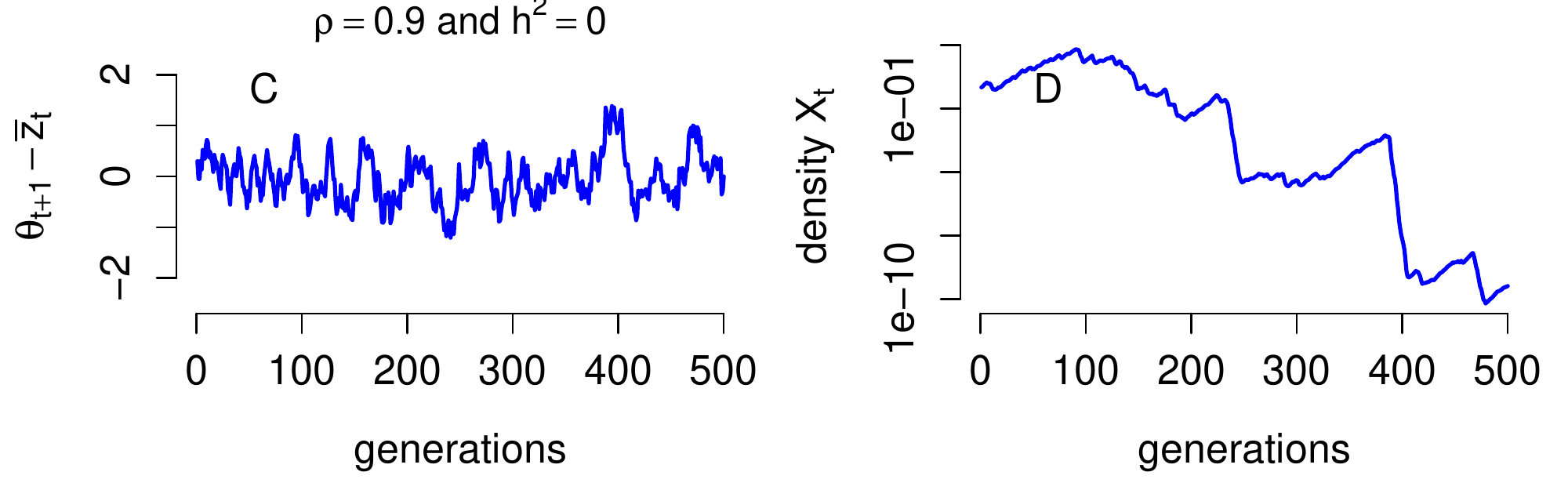}
\includegraphics[width=0.8\textwidth]{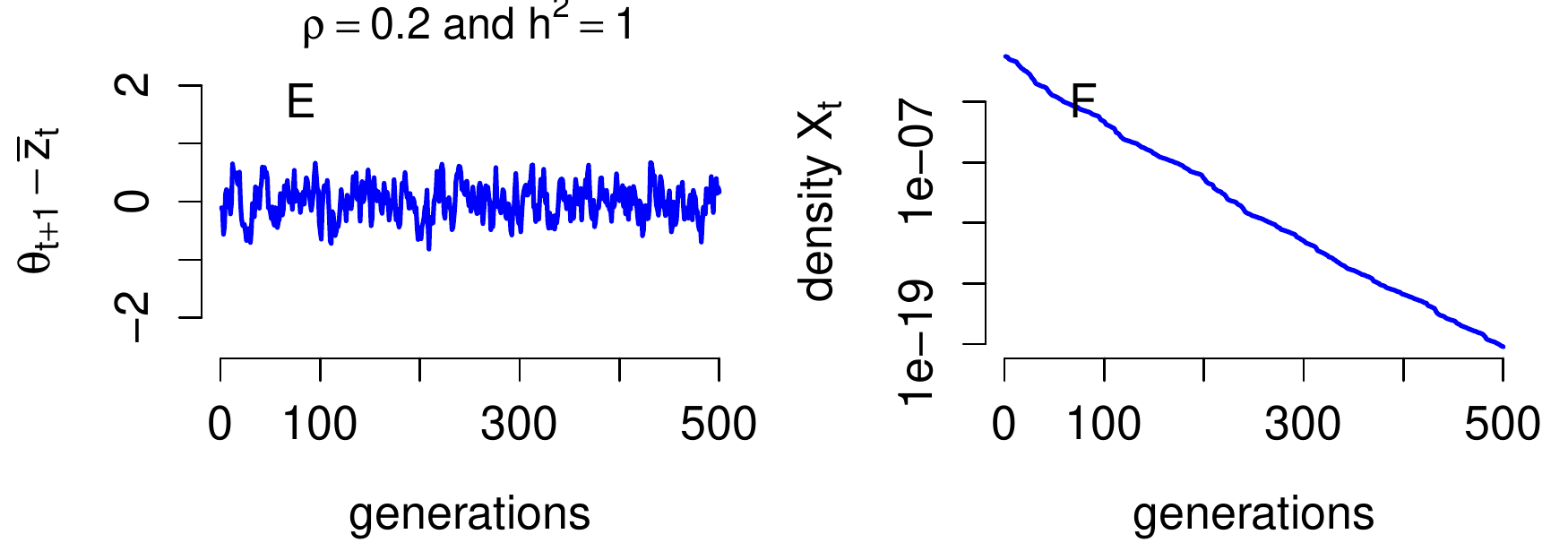}
\caption{\emph{Evolutionary rescue more likely in autocorrelated environments for \eqref{ee:one}--\eqref{ee:two}.} Left hand panels plot the difference between the mean population trait $\bar{z}_t$ and the optimal trait value $\theta_{t+1}$. The absolute value of this difference corresponds to the genetic load. Right hand panels plot the population density $X_t$. In (A)-(B), heritability and autocorrelation is high and the population persists i.e. $r(\mu)>0$. In (C)-(D), the population doesn't evolve (i.e. $h^2=0$) and autocorrelation is high, but the population tends toward extinction i.e. $r(\mu)<0$ and $\s_0$ is accessible. In (E)-(F), heritability is high but autocorrelation is low, but the population tends toward extinction.  Parameters:$\gamma=1$, $\sigma=0.5$,
$r_{\mbox{\tiny max}}=0.25$, $h^2, \rho$ as shown in panels, $\tau=0.5$, and
$a=0.001$}\label{fig:ER}
\end{figure}

Assume there exists a compact set $\s\subset [0,\infty)\times (-\infty,\infty)^2$ such that $Z_t$ enters and remains in $\s$ for sufficiently large $t$.  For example, this holds when $a(\bar{z})$ is a fixed positive constant. The extinction set $\s_0=\{z=(0,y)\in \s\}$ corresponds to extinction of the population. On this extinction set, the auxiliary dynamics $Y_t$ are a first-order, bivariate auto-regressive process:
\begin{equation}
\begin{pmatrix}\bar{z}_{t+1}\\ \bar{\theta}_{t+1}\end{pmatrix}=\underbrace{
\begin{pmatrix}
1-\alpha&\rho\alpha\\
0& \rho
\end{pmatrix}}_{=:A}\begin{pmatrix}\bar{z}_{t}\\ \bar{\theta}_{t}\end{pmatrix}+
\underbrace{\begin{pmatrix}\alpha \sqrt{1-\rho^2} \xi_{t+1}\\ \sqrt{1-\rho^2} \xi_{t+1}\end{pmatrix}}_{:=b_{t+1}}.
\end{equation}
As the spectral radius of $A$ equals $\max\{1-\alpha,\rho\}$ which is strictly less than one, the dynamics of $Y_{t+1}$ on the extinction set converge to a unique stationary  solution $\widehat Y_t= (\widehat {\bar z_t}, \widehat \theta_t)$ (see, e.g., \citet{diaconis-freedman-99}). The covariance matrix $\rm{Cov}[\widehat Y_t]$ of this stationary solution is given by (see, e.g., \citet{schreiber-moore-18})
\[
\rm{vec}(\rm{Cov}[\widehat Y_t])=(\rm{Id}-A\otimes A)^{-1} \rm{vec}(\rm{Cov}[b_{t}])
\]
where $\rm{vec}$ denotes the vec operator that concatenates the columns of the matrix into a column vector and $\otimes$ denotes the Kroenker product. The entries of this covariance matrix are given by
\begin{eqnarray*}
\rm{Var}[\widehat\theta_t ]&=&\tau^2,\\
\rm{Var}[\widehat {\bar z_t}]&=&\tau^2 \frac{\alpha(\alpha-1)\rho-\alpha}{(\alpha^2-3\alpha+2)\rho+\alpha-2}\mbox{, and }\\
\rm{Cov}[\widehat {\bar z_t},\widehat \theta_{t}]&=&\tau^2 \frac{\alpha}{(\alpha-1)\rho+1}.
\end{eqnarray*}
Furthermore, as $\E[\widehat \theta_t, \widehat \theta_{t+1}]=\rho V$, we have
\[
\rm{Cov}[\widehat {\bar z_t},\widehat \theta_{t+1}]=\tau^2 \frac{\rho\alpha}{(\alpha-1)\rho+1},
\]
and
\[
\rm{Var}[\widehat{\bar z}_{t}-\widehat \theta_{t+1}]=\frac{1-\rho}{1-(1-\alpha)\rho}\frac{2\tau^2}{2-\alpha}.
\]
Let $\mu$ be the law of the stationary solution $(0,\widehat Y_t)$ which is the unique ergodic probability measure on the extinction set $\s_0$. Then
\begin{equation}\label{eq:rmu}
r(\mu)
 =r_{\mbox{\tiny{max}}}-\frac{\gamma}{2}\left(\sigma^2+\rm{Var}[\widehat {\bar z_t}- \widehat{\theta}_{t+1}]\right)= r_{\mbox{\tiny{max}}}-\frac{\gamma}{2}\left(\sigma^2+\frac{1-\rho}{1-(1-\alpha)\rho}\frac{2\tau^2}{2-\alpha}\right).
\end{equation}
Theorem~\ref{thm:persist} implies if $r(\mu)>0$ then the population stochastically persists. Conversely if $r(\mu)<0$, then Theorem~\ref{thm:exclusion1} implies that the population goes extinct asymptotically with positive probability whenever $X_0$ is sufficiently low. Furthermore, if \[\P\left[\frac{\gamma (\widehat{\bar{z}}_t-\widehat{\theta}_{t+1})^2}{2\sigma^2\gamma+2}>r_{\mbox{\tiny{max}}}\right]>0\] and $r(\mu)<0$, then $\s_0$ is accessible and Corollary~\ref{cor1} implies that $\lim_{t\to\infty}X_t= 0$ with probability one.

\begin{figure}
\includegraphics[width=0.7\textwidth]{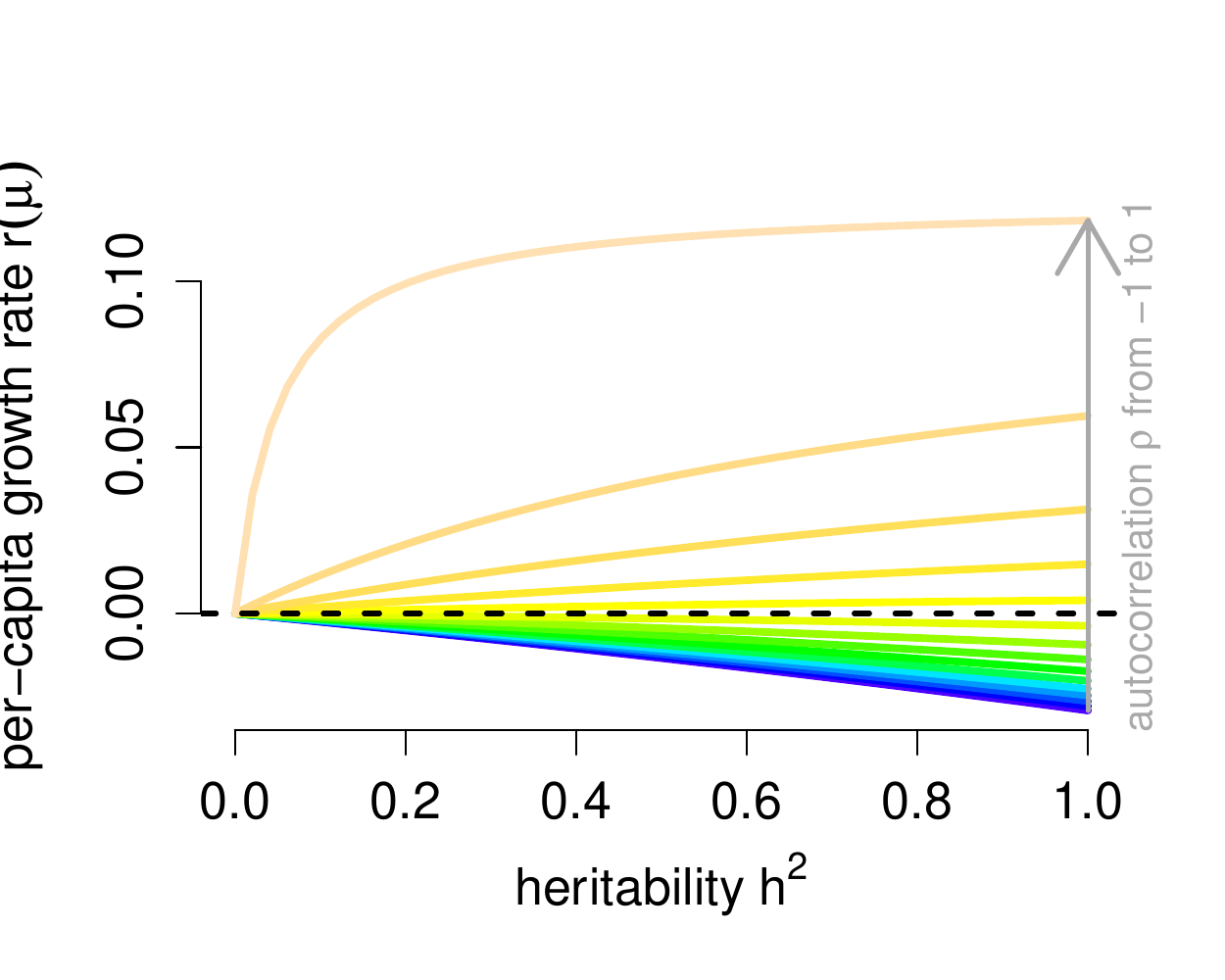}
\caption{The effects of heritability $h^2$ and temporal autocorrelation on the realized per-capita growth rate $r(\mu)$. With low autocorrelation, greater heritability decreases $r(\mu)$ and can lead to extinction. With higher autocorrelation, $r(\mu)$ increases with $h^2$ and fast evolution can rescue the population. Parameter values as in Figure~\ref{fig:ER}.}\label{fig:ER2}
\end{figure}

\emph{\textbf{Biological implications.} Equation~\eqref{eq:rmu} provides several insights about the effects of trait evolution in a fluctuating environment on population persistence. First, $r(\mu)$ decreases  with the strength $\gamma$ of stabilizing selection and the environmental variation $\tau^2$ in the optimal trait value. Intuitively, both of these quantities increase the long-term genetic load (i.e. the  ``reduction in the mean fitness of a population relative to a population composed entirely of individuals having optimal genotypes'' \citep{whitlock-davis-11}) and, consequently, the chance for persistence. Second, $r(\mu)$ increases with the temporal autocorrelation $\rho$ of the optimal trait in the changing environment. Intuitively, the stronger the temporal autocorrelation, the more time trait evolution has to track the optimal trait (compare Figs.~\ref{fig:ER}A,B to C,D).  Indeed, in the limit $\rho\uparrow1$, there is no effect of  $\tau^2$ on population persistence. Finally, the speed of evolution, which is determined by $\alpha$, plays an important role on persistence. In the limit of no evolution (i.e. $\alpha\downarrow 0$ ), the genetic variance and environmental variance contributing equally to a reduction in the realized per-capita growth rate i.e.
$
r(\mu)
 =r_{\mbox{\tiny{max}}}-\frac{\gamma}{2}\left(\sigma^2+2\tau^2\right)
$. Alternatively, if the speed of evolution is maximized (i.e. $\alpha\uparrow 1$), then in this limit $r(\mu)= r_{\mbox{\tiny{max}}}-\frac{\gamma}{2}\left(\sigma^2+2(1-\rho)\tau^2\right)$. Hence, if environmental fluctuations are nearly perfectly autocorrelated (i.e. $\rho\approx 1$), then trait evolution eliminates the environmental reduction in the per-capita growth rate. In contrast, if the environmental fluctuations are strongly negatively autocorrelated (i.e. $\rho\approx -1$), then fast evolution causes a twofold reduction in the realized per-capita growth rate. Intuitively, in negatively autocorrelated environments, passing on the optimal trait for the parents to the offspring provides the offspring with a suboptimal trait. Hence, in strongly negatively autocorrelated environments, non-evolving populations may be more likely to persist than evolving populations.} Several of these conclusions are illustrated in Fig.~\ref{fig:ER2}

\subsection{A structured SIS disease model}

To illustrate how auxiliary variables can account for population structure, we analyze a spatially structured disease model. A closely related model in continuous time was analyzed by \citet{benaim-strickler-19}. In this model, there are $k$ subpopulations of individuals with densities $S^i_t$ and $I^i_t$ of susceptible and infected individuals, respectively, in the $i$-th subpopulation. The subpopulation sizes are constant with a density of $N^i=S^i_t+I^i_t$ individuals in the $i$-th subpopulation. The length of time between updates corresponds to the infection duration. After this infectious period, infected individuals become susceptible again. Individuals in subpopulation $i$ experience a mortality rate of $\mu_{t+1}^i$ over the time interval $(t,t+1]$. Thus, the fraction of individuals surviving in this subpopulation is $\exp(-\mu_{t+1}^i)$. The susceptible individuals that survive in subpopulation $i$ become infected with a net force of infection $\sum_{j=1}^n\beta^{ij}I^j_t$ where $\beta^{ij}$ is the per-capita force of infection from subpopulation $j$ to subpopulation $j$. If infections are Poisson distributed, the fraction of susceptible individual escaping infection is $\exp(-\sum_{j=1}^k \beta^{ij}I^j_t)$. Under these assumptions, the SIS dynamics are given by the following $k$ equations:
\begin{equation}\label{eq:SIS1a}
I_{t+1}^i = (N^i-I^i_t)e^{-\mu_{t+1}^i}\left(1-\exp\left(-\sum_{j=1}^k \beta^{ij}I^j_t\right)\right) \quad i=1,2,\dots,k.
\end{equation}
To complete the model, we assume that the mortality rates are given by a first order auto-regressive process
\begin{equation}\label{eq:SIS1b}
\mu_{t+1}=A\mu_t+\xi_{t+1}
\end{equation}
where $\mu_t=(\mu_t^1,\dots,\mu_t^k)$ is the vector of mortality rates, $A$ is a non-negative matrix with a spectral radius $<1$ and $\xi_t$ are i.i.d. random, non-negative vectors taking values in a compact set.

To place \eqref{eq:SIS1a}-\eqref{eq:SIS1b} into our framework, define $X_t$ to be the total density of infected and $F^i_t$ to be the fraction of  total infected in subpopulation $i$:
\[
X_t=\sum_{i=1}^k I_t^i \mbox{ and }F^i_t = I^i_t/X_t.
\]
If we define $\phi(x)=(1-e^{-x})/x$, then for $X_t>0$ the dynamics of \eqref{eq:SIS1a}-\eqref{eq:SIS1b} are equivalent to 
\begin{equation}\label{eq:SIS2}
\begin{aligned}
X_{t+1}=&  \sum_{i=1}^k (N^i-F^i_tX_t) e^{-\mu_{t+1}^i}\left(\sum_{j=1}^k \beta^{ij} F_t^j X_t\right) \phi\left(\sum_{j=1}^k \beta^{ij}F_t^jX_t\right)\\
F_{t+1}^i=& (N^i-F^i_tX_t) e^{-\mu_{t+1}^i}\left(\sum_{j=1}^k \beta^{ij} F_t^jX_t\right) \phi\left(\sum_{j=1}^k \beta^{ij}F_t^jX_t\right)/X_{t+1}\\
\mu_{t+1}=&A\mu_t+\xi_{t+1}.
\end{aligned}
\end{equation}
As $\lim_{x\to 0}\phi(x)=1$, for $X_t=0$, the dynamics extend uniquely and continuously  to
\begin{equation}\label{eq:SIS3}
\begin{aligned}
X_{t+1}=&0\\
F_{t+1}^i=&N^i e^{-\mu_{t+1}^i}\left(\sum_{j=1}^n \beta^{ij} F_t^j\right)/\left(\sum_{j,\ell=1}^k  e^{-\mu_{t+1}^j}N^j \beta^{j\ell} F_t^\ell \right)\\
\mu_{t+1}=&A\mu_t+\xi_{t+1}.
\end{aligned}
\end{equation}
Thus,  \eqref{eq:SIS1a}-\eqref{eq:SIS1b} can be written in the desired form \eqref{eq:main} with $Y_t=(F_t,\mu_t)$. The state space for these dynamics is $\s=[0,\sum_i N^i]\times \Delta$ where $\Delta=\{x\in [0,1]^k:\sum_{i=1}^k x^i=1\}$ is the probability simplex. The extinction set is $\s_0=\{0\}\times \Delta.$

The $F_t$ dynamics on the extinction set $\s_0$ correspond to projection of the following stochastic linear difference equation onto the simplex (i.e. $F_t=V_t/\sum_i V^i_t$)
\begin{equation}\label{eq:SIS4}
V_{t+1}=\rm{diag}(e^{-\mu_{t+1}^i}N^i)B V_t
\end{equation}
where $B=(\beta^{ij})_{i,j}$ and $\rm{diag}(a_i)$ denotes a diagonal matrix with entries $a_1,\dots,a_n$ If $B$ is a primitive matrix, then the random Perron Frobenius Theorem (see, e.g., \citet[Proposition 8.3]{jmb-14}) implies that $F_t$ converges to a unique stationary distribution $\widehat{ F}$. Let $\mu$ be the product of the Dirac measure at $\{0\}$ and the law of $\widehat F$. Then $r(\mu)$ corresponds to the dominant Lyapunov exponent of \eqref{eq:SIS4}. If $r(\mu)>0$, Theorem~\ref{thm:persist} implies that the disease persists in probability and almost surely. Conversely if $r(\mu)<0$, then Theorem~\ref{thm:exclusion1} implies that the disease is unlikely to establish whenever arriving at low densities.

In general, finding explicit expression for this dominant Lyapunov exponent $r(\mu)$ is challenging. However, in the special case, of perfect mixing (i.e. $B=\beta J/n$ where $J$ is the matrix of ones and $\beta>0$), work of \citet{metz-etal-83} implies that
\begin{equation}\label{eq:SISr}
r(\mu)=\E\left[\log\left(\frac{1}{k}\sum_{i=1}^k \beta N^i e^{-\mu_t^i}\right)\right].
\end{equation}
In the special case of a single subpopulation (say $i$), $r(\mu)=\E[\log (\beta N^i e^{\mu_t^i})]$ and persistence requires $\log(\beta N^i)>\E[\mu_t^i]$. Even if persistence within any subpopulation isn't possible, the coupling of the subpopulations by infection can allow for persistence. To see why, applying Jensen's inequality to the inner and outer averages of \eqref{eq:SISr} and assuming there is variation in the $\mu_t^i$, we get
\begin{equation}\label{eq:jensen}
\log\left(\frac{1}{k}\sum_{i=1}^k \beta N^i \E\left[e^{-\mu_t^i}\right]\right)>r(\mu)>\frac{1}{k}\sum_{i=1}^k\E\left[\log\left( \beta N^i e^{-\mu_t^i}\right)\right].
\end{equation}
The second equality in this equation implies that $r(\mu)$ for the structured population is strictly greater than the averaged realized per-capita growth rate for the isolated subpopulations.

\emph{\textbf{Biological implications.} The first inequality in \eqref{eq:jensen} implies that disease persistence is possible for the entire population despite it being unable to persist within any subpopulation. In a fully mixed population this only occurs if at least one of arithmetically averaged reproductive numbers $\beta N_i \E[e^{-\mu_t^i}]$ is greater than one. However, even if  all the arithmetically averaged reproductive numbers are less than one, then an approximation developed in \citep{prsb-10} implies that imperfect mixing and positive temporal autocorrelations in the mortality rates $\mu_t^i$ can allow for disease persistence}.

\section{Discussion}\label{sec:end}

We developed criteria for stochastic persistence and exclusion for stochastic, discrete-time multispecies models that allow for internal and external variables. Our criteria are based on species realized per-capita growth rates when rare~\citep{turelli-81,chesson-82,chesson-ellner-89,chesson-94} evaluated at ergodic probability measures supporting a subset of species. For the non-supported species, these realized per-capita growth rates correspond to Lyapunov exponents~\citep{ferriere-gatto-95} that characterize the rate of growth of infinitesimal additions of the missing species. As pioneered by \citep{hofbauer-81} for ordinary differential equations, our stochastic persistence criterion requires the existence of positive weights associated with each species such that this weighted combination of realized per-capita growth rates, what we call the realized community per-capita growth rate, is positive for all ergodic probability measures supporting a subset of species. When this occurs, we show there is persistence from the ensemble point of view (i.e. persistence in probability) and the typical trajectory point of view (i.e. almost-sure persistence). These results extend prior results~\citep{jmb-11,jdea-11,jmb-14} for stochastic difference equations in four ways. First, we provide explicit lower bounds for the fraction of time trajectories spend near the boundary. Second, we prove persistence in probability for models with population structure without the additional accessibility assumption made in \citep{jmb-11} for unstructured population models. Moreover, we provide explicit lower bounds for the probability of any species being below a low density far into the future. Third, our results provide explicit estimates for the time required to escape a neighborhood of the extinction set. This time to escape decreases with the realized community per-capita growth rate.  Finally, our results account for more types of internal variables than the multispecies, non-linear matrix models of \citep{jmb-14} as we illustrated with our example of trait evolution.

No understanding of persistence is complete without sufficient conditions for asymptotic exclusion of one or more species. Prior results about exclusion for stochastic difference equations have been limited to scalar models~\citep{chesson-82,gyllenberg-etal-94b,jmb-14,jbd-14} or monotone competing species models in i.i.d. environments~\citep{chesson-ellner-89}. Here we introduced two types of extinction results. The first result is the natural complement of our persistence theorem: if the community per-capita growth rate is negative for all ergodic probability measures supporting a subset of species, then the trajectories of communities starting near the extinction set are likely to asymptotically converge to this extinction set. To our knowledge, this is the first time for stochastic models such a criterion has been proven. This criterion is most useful for showing the entire extinction set is a stochastic attractor, as illustrated with the evolutionary game of rock-paper-scissors. For many other systems, however, the extinction set is not a stochastic attractor despite between some species being extinction prone. For example, in our model of competing species, both species can persist on their own, hence, the extinction set can not be an attractor as the origin is a stochastic repellor. For these models, we proved a discrete-time analog of an extinction result due to \citet{hening-nguyen-18} for unstructured Kolmogorov stochastic differential equation models. This result requires identifying a subset of species which satisfy the Hofbauer condition for stochastic persistence but for which all missing species have negative realized per-capita growth rates. As with our persistence results, we provide explicit lower bounds for the probability of extinction in terms of the realized per-capita growth rates. 

While our persistence and exclusion results are applicable to many types of models, several major mathematical challenges remain. First and foremost, as with the deterministic theory, there is a gap between the conditions for stochastic persistence (Theorem~\ref{thm:persist}) and exclusion (Theorems~\ref{thm:exclusion1} and \ref{thm:exclusion2}). In the deterministic persistence theory, this gap has been partially closed by considering Morse decomposition of the extinction set (i.e. the dynamics are gradient-like when collapsing a finite number of invariant sets to points) and allowing the species weights $p^i>0$ to be chosen separately for each of the invariant sets of the Morse decomposition~\citep{jde-00,garay-hofbauer-03,jde-10,jmb-18}. Having a similar result for stochastic models would further increase the applicability of the persistence and exclusion theory. Second, our results assume the dynamics remain on a compact set. While this assumption is biologically realistic, theoretical population biologists do build models that violate this assumption e.g. using lognormal random variables on $[0,\infty)^n$. Hence, extending our results to this setting would be useful and should be possible with the methods developed by \citet{hening-nguyen-18,benaim-18}. Third, often evaluating the persistence or extinction criteria for particular models can be exceptionally challenging. Hence, it will be useful to identify special classes of models beyond the stochastic Lotka-Volterra difference equations for which they can be explicitly evaluated. Moreover, the further development of approximation methods, such as the small-noise approximations of \citep{chesson-94}, is needed. Fourth, in deterministic theory, coexistence can be equated with the existence of a positive attractor~\citep{jtb-06} that need not be globally attracting i.e. need not be a permanent system~\citep{hofbauer-sigmund-98}. Developing methods to identify when there are positive stochastic attractors potentially coexisting with stochastic attractors on the extinction set will require going beyond realized per-capita growth rates (see, e.g., \citet{jbd-14} for single species models).

\vskip 0.1in
\noindent\textbf{Acknowledgments.} This work was supported in part by United States National Science Foundation Grant DMS-1716803 to SJS and Swiss National Science Foundation Grant 2000021−175728 to MB.

\section{Proofs}
\subsection{Proof of Proposition~\ref{prop:invasion}\label{proof:invasion}}
We begin with the following fundamental lemma. Recall $\xi_t$ take values in a polish space $\Xi$ and have a common law $m(d\xi)$.
\begin{lem}
\label{th:birk2}
Let $g : \s \times \Xi \mapsto \R$ be a measurable map such that $\sup_{z\in \s} \int g(z,\xi)^2 m(d\xi)<\infty. $ Define $\bar{g}(z) = \int g(z,\xi) m(d\xi).$ Then
\begin{itemize}
\item[(i)] For all $z \in \s$ and $Z_0 = z$
\[\lim_{t \to  \infty} \frac{\sum_{s = 0}^{t-1} g(Z_s,\xi_{s+1}) -
 \sum_{s = 0}^{t-1} \bar{g}(Z_s)}{t} = 0 \mbox{
with probability one.}\]

\item[(ii)] Let $\mu$ be an  invariant (respectively ergodic) probability measure for $(Z_t)$, then there exists a bounded measurable map $\hat{g}$ such that with probability one and for $\mu$-almost every $z$ \[\lim_{t \to  \infty} \frac{\sum_{s = 0}^{t-1} g(Z_s,\xi_{s+1})}{t} = \lim_{t \to \infty}  \frac{\sum_{s = 0}^{t-1} \bar{g}(Z_s)}{t} = \hat{g}(z)
\mbox{ when }Z_0=z.\]
 Furthermore $$\int \bar{g}(z) \mu(dz) = \int \hat{g}(z) \mu(dz) \mbox{ (respectively }
\hat{g}(z) = \int \bar{g}(z) \mu(dz) \quad \mu-\mbox{almost surely} ).$$
\end{itemize}
\end{lem}
\begin{proof} To prove the first assertion, let $C=\sup_{z\in \s}\int g(z,\xi)^2m(d\xi)$ and define the martingale
\[
M_t=\sum_{s=0}^{t-1}\left(g(Z_s,\xi_{s+1})-\bar{g}(Z_s)\right).
\]
with respect to  the sigma-algebra $\mathcal F_t$ generated by $(Z_0,\xi_0),\dots,(Z_t,\xi_t)$. $M_t$ has square integrable martingale differences as $\E[(M_t-M_{t-1})^2|\mathcal F_{t-1}]\le 2C^2.$ Define the previsible increasing process $\langle M\rangle_t$ by $\langle M\rangle_0=0$ and $\langle M\rangle_t=\langle M\rangle_{t-1}+\E[(M_t-M_{t-1})^2|{\mathcal F}_{t-1}]$. $\langle M\rangle_t$ is known as the angle-brackets process \citep[Section 12.12]{williams-91}. By construction, we have $\langle M\rangle_t\le 2C^2 t$. Define $\langle M\rangle_\infty =\lim_{t\to\infty} \langle M\rangle_t$, which exists (possibly infinite) as $\langle M\rangle_t$ is increasing. On the event $\langle M\rangle_\infty<+\infty$, \citet[Theorem 12.13a]{williams-91} implies that $\lim_{t\to\infty}M_t$ exists and is finite. In particular, $\lim_{t\to\infty}M_t/t=0$ on the event $\langle M\rangle_\infty<+\infty.$ On the event $\langle M\rangle_\infty=+\infty$, \citet[Theorem 12.14a]{williams-91} implies that $\lim_{t\to\infty}M_t/\langle M\rangle_t=0$. In particular, as $\langle M\rangle_t\le 2C^2t$, $\lim_{t\to\infty}M_t/t=0$ on the event $\langle M\rangle_\infty=+\infty$. Thus, $\lim_{t\to\infty}M_t/t=0$ with probability one which completes the proof of the first assertion. 

The second assertion follows from  Birkhoff's ergodic theorem applied to stationary Markov Chains (see \cite{meyn-tweedie-09}, Theorem 17.1.2)
\end{proof}

Equation \eqref{eq:r} and the second assertions of Proposition~\ref{prop:invasion}  follow directly from the preceding  lemma applied to $g(z,\xi) = \log f^i(z,\xi)$.  Assumption \textbf{A4} implies that $\log f^i(z,\xi_t)$ are uniformly integrable (UI). Therefore, equation \eqref{eq:r2} follows from \eqref{eq:r} and the UI convergence theorem. The first half of the third assertion follows from ergodicity. To prove the second half of the third assertion, let $i\in S(\mu)$. By assertion $(i)$ of Proposition~\ref{prop:invasion}
\[\lim_{t \rightarrow \infty} \frac{\log X^i_t}{t} = \hat r^{i}(z) \mbox{ where } Z_0=z\] for $\mu$-almost $z \in \{(x,y)\in \s: x^i>0\}.$
Let $\s^{i,\eta} = \{(x,y) \in \s \: : x^i \geq \eta\}$ and $\eta^*>0$ be such that $\mu(\s^{i,\eta})>0$ for all $\eta\le \eta^*$. By Poincar\'e Recurrence Theorem, for $\mu$ almost all $z \in \s^{i,\eta}$
\[\P_x[ Z_t \in \s^{i,\eta} \mbox{ infinitely often } ] = 1\]
for $\eta\le\eta^*$.
Thus $\widehat{r}_{i}(z) = 0$ for $\mu$-almost all $z \in \s^{i,\eta}$ with $\eta\le \eta^*$.
Hence $\widehat{r}^{i}(z) = 0$ for $\mu$-almost all $z \in \bigcup_{n \in \mathbf{N}} \s^{i,1/n} = \{(x,y) \in \s \: : x^i > 0\}.$ This proves assertion $(iii).$

\subsection{Proof of Proposition~\ref{prop:key}\label{proof:key}}
We begin by showing there is $T\ge 1$ and $\alpha\in (0,1)$ such that inequality \eqref{eq:key1} holds i.e. $P_TV(z)-V(z)\le -\alpha$ for all $z\in \s_0$. Suppose to the contrary. Then there exists an increasing sequence of times $t_k\uparrow \infty$, a decreasing sequence of positive reals $\alpha_k\downarrow 0$ and a sequence of points $z_k$ in $\s_0$ such that $P_{t_k}V(z_k)-V(z_k)\ge -\alpha_k$ for all $k$. Define a sequence of Borel  probability measures $\mu_k$ on $\s_0$ by
\[
\int h(z) \mu_k(dz)= \frac{1}{t_k}\E_{z_k}\left[\sum_{\tau=0}^{t_k-1} h(Z_\tau)\right]
\mbox{ for any continuous $h:\s_0\to\R$.}
\]
Let $\mu$ be a weak* limit point of the $\mu_i$, which exists as $\s_0$ is compact. By a standard argument due{} to Khasminskii (see, e.g.,  \citet[Theorem 1.1]{kifer-88}), $\mu$ is an invariant measure for the Markov chain and by weak* compactness, $\sum_i p^i r^i(\mu)\le 0$. By the ergodic decomposition theorem, there exists an ergodic measure $\nu$ supported on $\s_0$ such that inequality \eqref{eq:criterion} is violated; a contradiction.

Let $\phi(x)=e^x-1-x$. For any real $C$,
\[
|\phi(-\theta C)|\le \sum_{k\ge 2} \frac{|\theta C|^k}{k!}
\le \theta^2 e^C \mbox{ whenever $|\theta|\le 1$.}
\]
Our assumption that $\sup_{z,\xi}|\log f^i(z,\xi)|<\infty$ implies that there exists $C>0$ such that
\[|\sum_{\tau=0}^{T-1}\sum_i p^i \log f^i (z_\tau,\xi_{\tau})|\le C\] for any $z_0,\dots,z_{T-1}\in \s_0$ and $\xi_1,\dots,\xi_T\in \Xi.$ Thus, for $z\in \s\setminus \s_0$
\begin{eqnarray*}
P_T V_\theta(z)&=& \E_z\left[\exp(-\theta \sum_i p^i \log X_T^i) \right]\\
&=& \E_z\left[\exp\left(-\theta\sum_i p^i \left( \sum_{\tau=0}^{T-1}\log f^i(Z_\tau,\xi_{\tau+1})+\log x^i\right)\right)     \right]\\
&=& \E_z\left[\exp\left(-\theta\sum_i p^i \sum_{\tau=0}^{T-1}\log f^i(Z_\tau,\xi_{\tau+1}\right)\right]V_\theta(z)\\
&=& \E_z\left[1-\theta \sum_i p^i \sum_{\tau=0}^{T-1}\log f^i(Z_\tau,\xi_{\tau+1})+\phi\left(-\theta \sum_i p^i \sum_{\tau=0}^{T-1}\log f^i(Z_\tau,\xi_{\tau+1})\right) \right]V_\theta(z)\\
&\le & (1-\theta\alpha+\theta^2e^C )V_\theta(z).
\end{eqnarray*}
Choosing $\theta=\alpha\,e^{-C}/2$ and $\rho=1-\theta\alpha/2$ completes the proof of the proposition.

\subsection{Proof of Theorem~\ref{thm:persist}\label{proof:persist}}

In this section, we provide the details of the proof of Theorem~\ref{thm:persist} that are not presented in the main text. Through out this section, we assume that \eqref{eq:criterion} holds. Let $\theta>0$, $T\ge 1$, $\rho\in (0,1)$, and $\beta>0$ be as given in Proposition~\ref{prop:key}.

We begin proving persistence in probability for the general case of $T\ge 1$. Given any $z_0\in \s\setminus \s_0$, any integer $t=\ell T+s\ge 1$ with $\ell\ge 0$ and $0\le s\le T-1$, and $\eta>0$, Proposition~\ref{prop:key} implies
\begin{eqnarray*}
\P_{z_0}\left[Z_{t}\in \s_\eta\right]\min_{z\in \s_\eta\setminus\s_0}V_\theta(z)
&\le& \int_\s  V_\theta(z)(\delta_{z_0}P_{\ell T+s})(dz)\\
&\le &\rho^\ell \int_\s V_\theta(z)
(\delta_{z_0}P_s)(dz)+\frac{\beta}{1-\rho}\\
&\le& \rho^\ell \max_{0\le \tau \le T-1}\int_\s V_\theta(z)
(\delta_{z_0}P_\tau)(dz)+\frac{\beta}{1-\rho}
\end{eqnarray*}
As $\rho^\ell\to 0$ as $t\to \infty$, inequality \eqref{eq:lb} implies 
\[
\limsup_{t\to\infty}\P_{z_0}[Z_t\in \eta]\le a \eta^b  \mbox{ where }a=\frac{\beta}{a_0(1-\rho)}
\]
which completes the proof of persistence in probability.

To complete the proof of almost-sure persistence presented in the main text, we need the following lemma which follows the strategy of proof of \citet[Lemma 6]{jmb-11}.

\begin{lem}\label{lem:pi} For all $z\in \s\setminus \s_0$, the weak* limit points $\mu$ of $\Pi_t$ almost surely satisfy $\mu(\s_0)=0$.
\end{lem}

\begin{proof}
 The process $\{Z_t\}_{t=0}^\infty$ being a (weak) Feller Markov chain over a compact set $\s$ implies that the set of weak* limit points of $\{\Pi_t\}_{t=0}^\infty$ is almost surely a non-empty compact subset of probability measures supported by $\s.$ Almost sure invariance of these weak* limit points follows from Lemma~\ref{th:birk2} $(i)$.

Assertion $(i)$ of  Lemma \ref{th:birk2} applied to $g(z,\xi) = \sum_i p^i\log f^i(z,\xi)$ gives
we have
\begin{eqnarray*}
\lim_{t \to \infty} \frac{\sum_i p^i\log X^i_t-\sum_ip^i\log x^i  - \sum_{s=0}^{t-1}\sum_ip^i \overline{\log f^i}(Z_s )}{t} = 0
\end{eqnarray*}
where $\overline{\log f^i}(z)=\int \log f^i (z,\xi)m(d\xi).$
 Since $\limsup_{t\to\infty}\frac{1}{t} \sum_ip^i\left(\log X^i_t-\log x^i\right) \le 0$ almost surely, we get that
 \begin{equation}
 \label{eq:noname}
 \sum_ip^ir^i(\mu) \le 0
 \end{equation}
 almost surely for any weak* limit point $\mu$ of $\{\Pi_t\}_{t=0}^\infty$.

 Since $\s_0$ and $\s\setminus\s_0$ are invariant, there exists $\alpha\in (0,1]$ such that $\mu=(1-\alpha)\nu_0+\alpha \nu_1$ where $\nu_0$ is an invariant probability measure with $\nu_0( \s_0)=1$ and $\nu_1$ is an invariant probability measure with $\nu_1(\s_0)=0$. By Proposition \ref{prop:invasion},  $\sum_ip^ir^i(\nu_1) = 0$. Thus,  $(1-\alpha) \sum_ip^ir^i(\nu_0) \leq 0$. Since by assumption $\sum_ip^ir^i(\nu_0) > 0$, $\alpha$ must be $1.$

\end{proof}

It follows that with probability one, the weak* limit points $\mu$ of $\Pi_t$ (given $Z_0=z\in \s\setminus\s_0$) are invariant probability measures satisfying $\mu(\s_0)=0$. To complete, the proof of almost-sure persistence, we need to provide uniform upperbounds to the amount of weight that invariant probability measures $\mu$ with $\mu(\s_0)=0$ place near $\s_0$. To this end, let $\mu$ be an invariant probability measure with $\mu(\s_0)=0$. As, in general, we can not assume that $\int V_\theta(z)\mu(dz)$ is well-defined and finite, we present a slightly longer argument than shown in the main text, to deal with this issue. This argument follows \citet[Proposition 4.24]{harier-18}. Let $M>0$ be any positive real. For two real numbers $a,b$, let $a\land b$ denote $\min\{a,b\}.$ Then invariance, Jensen's inequality and  Proposition~\ref{prop:key} imply
\[
\begin{aligned}
\int (V_\theta\land M)(z) \mu(dz)=&\int (P_{T}(V_\theta\land M))(z) \mu(dz) \\
=&\int \E_z[V_\theta(Z_T)\land M]\mu(dz) \le \int \E_z[V_\theta(Z_T)]\land M \mu(dz)\\
\le & \int (\rho V_\theta(z)+\beta)\land M \mu(dz).
\end{aligned}
\]
Iterating this inequality $k\ge 1$ times yields 
\begin{equation}\label{eq:addon}
\int (V_\theta\land M)(z) \mu(dz)\le \int ((\rho)^k V_\theta(z)+\beta/(1-\rho))\land M \mu(dz).
\end{equation}
By the dominated convergence theorem, taking the limit $k\to\infty$  yields 
\begin{equation}\label{eq:addon-2}
\int (V_\theta\land M)(z) \mu(dz)\le \frac{\beta}{1-\rho}.
\end{equation}
By the dominated convergence theorem, taking the limit $M\to\infty$ yields 
\begin{equation}\label{eq:addon-3}
\int V_\theta(z) \mu(dz)\le \frac{\beta}{1-\rho}.
\end{equation}
Thus, as in the main text, for any $\eta>0$,
\[
\mu(\s_\eta)\min_{z\in \s_\eta}V_\theta(z)\le\int V_\theta(z) \mu(dz) \le \frac{\beta}{1-\rho}.
\]
Inequality \eqref{eq:lb} implies that
\[
\mu(\s_\eta)\le a (\eta)^b \mbox{ for all }\eta\le 1 \mbox{ where }a=\frac{\beta}{a_0(1-\rho)}
\]
which completes the proof of almost-sure persistence.

\subsection{Proof of Theorem~\ref{thm:exclusion1} and Corollary~\ref{cor1}\label{proof:exclusion1}}
The strategy of this proof is based on the proof of Theorem 3.3 by \citet{benaim-lobry-16}. Define $V(x,y)=\sum_i p^i \log x^i $. As in the proof of Theorem~\ref{thm:persist}, assumption~\eqref{eq:exclusion1} and Proposition~\ref{prop:invasion} imply there exists $T\ge 1$, $\eta>0$ and $\alpha>0$ such that
\begin{equation}\label{eq:ex1a}
P_T V(z)-V(z)\le -\alpha \mbox{ for all }z\in \s_\eta.
\end{equation}
Moreover, using the same argument as found in Proposition~\ref{prop:key}, there exists $\rho\in(0,1)$ and $\theta>0$ such that (choosing $\eta>0$ to be smaller if necessary)
\begin{equation}\label{eq:ex2}
P_T V_\theta(z)\le \rho V_\theta(z)\mbox{ for all }z\in \s_\eta \mbox{ where }V_\theta(z)=\exp(\theta V(z)).
\end{equation}
Define the stopping time $\tau=\inf \{k: Z_{kT}\notin \s_\eta\}$ and the event $\mathcal{A}=\{\limsup_{t\to\infty} V(Z_t)/t\le -\alpha\}$. We will show that there exists a function $q(\varepsilon)\in (0,\eta)$ such that $\P_z(\mathcal{A})\ge q(\varepsilon)$ for $z\in \s_\varepsilon$ and $q(\varepsilon)\uparrow 1$ as $\varepsilon \downarrow 0.$ To this end, define $W_k=V_\theta(Z_{kT})$. Equation\eqref{eq:ex2} implies that $W_{k\land \tau}$ is a super martingale. Hence, if we define  $C(\varepsilon)=\sup_{z\in \s_\varepsilon} \exp(V(z))$ for any $\varepsilon>0$ (note that $C(\varepsilon)\downarrow 0$ as $\varepsilon\downarrow 0$), then
\[
\E\left[ W_{k\land\tau}\mathbf{1}_{\tau<\infty}\right]\le W_0 \le C(\varepsilon)^\theta \mbox{ whenever }Z_0=z\in \s_\varepsilon.
\]
Taking the limit as $k\to\infty$ and defining $D=\min_{z\in \s\setminus \s_\eta} \exp(V(z))>0$, the dominated convergence theorem implies that
\[
\P_z[\tau<\infty]D^\theta \le \E_z[W_\tau \mathbf{1}_{\tau<\infty}]\le C(\varepsilon)^\theta \mbox{ whenever }z\in \s_\varepsilon.
\]
Thus,
\[
\P_z[\tau=\infty]\ge 1-\left(\frac{C(\varepsilon)}{D}  \right)^\theta=:q(\varepsilon)\mbox{ whenever }z\in \s_\varepsilon
\]
and where $q(\varepsilon)\uparrow 1$ as $\varepsilon\downarrow 0.$

Next, consider the martingale, $M_n=\sum_{\ell=1}^n V(\ell T)-P_TV((\ell-1)T)$. By the strong law for martingales and inequality \eqref{eq:ex1a}, $\limsup_{n\to\infty} V(nT)/n\le -\alpha$ on the event $\tau=\infty.$ As we have assumed that $|\log f^i|$ is uniformly bounded on $\s\times \Xi$, $\limsup_{t\to\infty}V(t)/t \le -\alpha$ on the event $\tau=\infty$. Thus, as claimed, we have shown that
\[
\P_z[\mathcal{A}]\ge q(\varepsilon) \mbox{ for }z\in \s_\varepsilon.
\]
As $V(Z_t)\downarrow -\infty$ implies that $\rm{dist}(Z_t,\s_0)\downarrow 0$, the proof of Theorem~\ref{thm:exclusion1} is complete.

To prove Corollary~\ref{cor1}, define the event \[\mathcal{E}=\left\{\lim_{t\to\infty} \rm{dist}(Z_t,\s_0)=0\right\}.\]
Choose $\varepsilon>0$ such that $q(\varepsilon)>1/2.$
Define the stopping time
\[
{\widetilde\tau} = \inf \{t\ge 0 \ : Z_t \in \s_\epsilon\}.
\]  Since $\s_0$ is accessible, there exists $\gamma >0$ such that $\P_{z}[{\widetilde\tau}<\infty]>\gamma$ for all $z\in \s$.  The strong Markov property implies that for all $z\in \s$
\begin{eqnarray*}
\P_z\left[ \mathcal{E}\right] &=& \E_{z}\left[ \P_{Z_{\widetilde\tau}} \left[ \mathcal{E} \right] \mathbf{1}_{\{ {\widetilde\tau}<\infty\}} \right]+\E_{z}\left[ \P_{Z_{\widetilde\tau}} \left[ \mathcal{E} \right] \mathbf{1}_{\{ {\widetilde\tau}=\infty\}} \right]\\
&=& \E_{z}\left[ \P_{Z_{\widetilde\tau}} \left[ \mathcal{E} \right] \mathbf{1}_{\{ {\widetilde\tau}<\infty\}} \right]\ge \gamma/2.
\end{eqnarray*}

Let $\mathcal{F}_t$ be the $\sigma$-algebra generated by $\{Z_1,\dots,Z_t\}$.
The L\'{e}vy zero-one law implies that for all $z\in\s$, $\lim_{t\to \infty} \E_{z}\left[ \mathbf{1}_{\mathcal{E}} | \mathcal{F}_t\right]= \mathbf{1}_{\mathcal{E}}$ almost surely. On the other hand, the Markov property implies that $\E_{z}\left[ \mathbf{1}_{\mathcal{E}} | \mathcal{F}_t\right]=\E_z[\P_{Z_t}[{\mathcal E}]]\ge \gamma/2$ for all $z\in\s$. Hence $\P_{z}[\mathcal{E}]=1$ for all $z\in\s$.

\subsection{Proof of Theorem~\ref{thm:exclusion2}\label{proof:exclusion2}}
This proof follows the strategy of Theorem~\ref{thm:exclusion1}, but using a $V$ function introduced by \citet{hening-nguyen-18}.

Let $I\subset\{1,2,\dots,n\}$ be the subset of species and $\{p^i\}_{i\in I}$ the set of positive reals such that $\sum_{i\in I} p^i r^i(\mu)>0$ for every ergodic $\mu$ supported on $\s^I_0=\{z=(x,y)\in \s^I: \prod_{i\in I}x^i=0\}$ where $\s^{I}=\{z=(x,y)\in \s: x^i=0$ for all $i\notin I\}.$ Assume that $r^i(\mu)<0$ for all $i\notin I$ and ergodic $\mu$ such that $\mu(\s^I_+)=1$ where $\s^I_+=\s^I\setminus \s^I_0.$ Choose $\delta>0$ and $\alpha>0$ such that $-\sum_{i\in I} p^i r^i(\mu)+\delta \max_{i\notin I} r^i(\mu)\le -2\alpha$ for all ergodic probability measures $\mu$ such that $\mu(\s^I)=1.$

Define $V(x,y)=-\sum_{i\in I} p^i \log x^i+\delta \max_{i\notin I}\log x^i $ and 
\[ V_\theta(x,y)=e^{\theta V(x,y)}=\left(\prod_{i\in I}(x^i)^{-\theta p^i}\right) \max_{i\notin I}(x^i)^{\theta \delta}.\] As in the proof of Theorem~\ref{thm:persist} and Proposition~\ref{prop:invasion}, there exists $T\ge 1$, $\eta\in (0,1]$ such that
\begin{equation}\label{eq:ex1}
P_T V(z)-V(z)\le -\alpha \mbox{ for all }z\in K_\eta
\end{equation}
where $K_\eta=\{(x,y)\in \s: \max_{i\notin I} x^i \le \eta\}.$
Moreover, using the same argument as found in Proposition~\ref{prop:key}, there exists $\rho\in(0,1)$ and $\theta>0$ such that (choosing $\eta>0$ to be smaller if necessary)
\begin{equation}\label{eq:ex3}
P_T V_\theta(z)\le \rho V_\theta(z)\mbox{ for all }z\in K_\eta\setminus\s^I_0.
\end{equation}
Choose $\widetilde\eta>0$ such that $\{z: V_\theta(z)\le \widetilde \eta\}\subset K_\eta.$ Define the stopping time $\tau=\{k: V_\theta(Z_{kT})\ge \widetilde \eta\}$ and the event $\mathcal{A}=\{\limsup_{t\to\infty} V(Z_t)/t\le -\alpha\}$. Define $W_k=V_\theta (Z_{kT})$ and $\widetilde W_k = \widetilde \eta \land W_k$. Equation~\eqref{eq:ex3} implies that $W_{k\land \tau}$ is a super martingale. Thus, for any $k$, concavity of $t\mapsto \delta \land t$ and Jensen's inequality implies that
\[
\begin{aligned}
\E\left[ \widetilde W_{k\land\tau}\mathbf{1}_{\tau<\infty}\right]
\le& \widetilde \eta\land \E\left[ W_{k\land\tau}\mathbf{1}_{\tau<\infty}\right] \\ \le&
\widetilde \eta \land W_0 = \widetilde \eta\land \frac{\max_{i\notin I}(x^i)^{\theta\delta}}{\prod_{i\in I}(x^i)^{\theta p^i}} =:C(z)\\
&  \mbox{ whenever }Z_0=z=(x,y)\in K_\eta\setminus\s^I_0.
\end{aligned}
\]
Taking the limit as $k\to\infty$, the dominated convergence theorem implies that
\[
\P_z[\tau<\infty]\widetilde \eta \le \E_z[\widetilde W_\tau \mathbf{1}_{\tau<\infty}]\le C(z)\mbox{ for all }z\in K_\eta\setminus\s^I_0.
\]
Thus,
\[
\P_z[\tau=\infty]\ge 1-\frac{C(z)}{\widetilde \eta}\mbox{ for all }z\in K_\eta\setminus\s^I_0.
\]

Next, consider the martingale, $M_n=\sum_{\ell=1}^n V(\ell T)-P_TV((\ell-1)T)$. By the strong law for martingales and \eqref{eq:ex1}, $\limsup_{n\to\infty} V(nT)/n\le -\alpha$ on the event $\tau=\infty.$ As we have assumed that $|\log f^i|$ is uniformly bounded on $\s\times\Xi$, $\limsup_{t\to\infty}V(t)/t \le -\alpha$ on the event $\tau=\infty$. Thus, we have shown that
\[
\P_z[\mathcal{A}]\ge 1-\frac{C(z)}{\widetilde \eta} \mbox{ for all }z\in K_\eta\setminus\s^I_0.
\]
Let $M=\sup_{(x,y)\in \s} \max_i x^i>0.$ As $\limsup_{t\to\infty} -\frac{\log X^i_t}{t}\ge \limsup_{t\to\infty} -\frac{\log M}{t} =0$ with probability one whenever $X_0^i>0$,
\[
\begin{aligned}
-\alpha\ge \limsup_{t\to\infty}\frac{V(Z_t)}{t}=&\limsup_{t\to\infty}\frac{1}{t}\left( -\sum_{i\in I} p^i \log X_t^i +
\delta \max_{i\notin I}\log X_t^i \right)\\
\ge & \limsup_{t\to\infty}\frac{\delta}{t} \max_{i\notin I}\log X_t^i
\end{aligned}
\]
almost surely on the event $\mathcal{A}$ whenever $Z_0=z\in K_\eta\setminus \s_0^I$. Hence, on this event, $\lim$ $\rm{dist}(Z_t,\s_0^I)\downarrow 0$ and the proof of Theorem~\ref{thm:exclusion2} is complete.

\bibliography{feedbacks}

\end{document}